\documentclass[10pt,journal,compsoc]{IEEEtran}
\pdfoutput=1 
\ifCLASSOPTIONcompsoc
\usepackage[nocompress]{cite}
\else
  \usepackage{cite}
\fi
\usepackage{algorithm}
\usepackage{algpseudocode}
\usepackage{epsfig}
\usepackage{graphicx}
\usepackage{amsfonts}
\usepackage{comment}
\usepackage{caption}
\usepackage{color}
\usepackage{amssymb}
\usepackage{amsmath}
\usepackage{amsthm,mathtools}
\usepackage{multirow}
\usepackage{etoolbox}
\usepackage{balance}
\usepackage{booktabs}
\usepackage{enumitem}
\usepackage[table]{xcolor}
\usepackage{hyperref}
\usepackage{tikz}
\DeclarePairedDelimiter{\ceil}{\lceil}{\rceil}
\usepackage{caption}
\usepackage{boldline}
\usepackage{makecell}
\setcellgapes{1pt}

\renewcommand{\baselinestretch}{0.9}

\newtheorem{lemma}{Lemma}
\AfterEndEnvironment{lemma}{\noindent\ignorespaces}

\newtheorem{definition}{Definition}[section]
\AfterEndEnvironment{definition}{\noindent\ignorespaces}
\newtheorem{example}{Example}[section]
\AfterEndEnvironment{example}{\noindent\ignorespaces}

\ifCLASSINFOpdf
\else
\fi
\begin{document}

\title{CAHPHF: Context-Aware Hierarchical QoS Prediction with Hybrid Filtering}

\author{Ranjana~Roy~Chowdhury,
        Soumi~Chattopadhyay,~\IEEEmembership{Member,~IEEE,} and
        Chandranath~Adak,~\IEEEmembership{Member,~IEEE} 

\IEEEcompsocitemizethanks{\IEEEcompsocthanksitem R. R. Chowdhury and S. Chattopadhyay are with Department of CSE, Indian Institute of Information Technology Guwahati, India-781015.  

C. Adak is with Centre for Data Science, JIS Institute of Advanced Studies and Research, JIS University, Kolkata, India-700091.

e-mail: $soumi@iiitg.ac.in$

This work has been submitted to the IEEE for possible publication. Copyright may be transferred without notice, after which this version may no longer be accessible.}

}

\makeatletter
\long\def\@IEEEtitleabstractindextextbox#1{\parbox{0.922\textwidth}{#1}}
\makeatother

\IEEEtitleabstractindextext{
\begin{abstract}
With the proliferation of Internet-of-Things and continuous growth in the number of web services at the Internet-scale, the service recommendation is becoming a challenge nowadays. One of the prime aspects influencing the service recommendation is the Quality-of-Service (QoS) parameter, which depicts the performance of a web service. In general, the service provider furnishes the value of the QoS parameters during service deployment. However, in reality, the QoS values of service vary across different users, time, locations, etc. Therefore, estimating the QoS value of service before its execution is an important task, and thus the QoS prediction has gained significant research attention. Multiple approaches are available in the literature for predicting service QoS. However, these approaches are yet to reach the desired accuracy level. In this paper, we study the QoS prediction problem across different users, and propose a novel solution by taking into account the contextual information of both services and users. Our proposal includes two key steps: (a) hybrid filtering and (b) hierarchical prediction mechanism. 
On the one hand, the hybrid filtering method aims to obtain a set of similar users and services, given a target user and a service. On the other hand, the goal of the hierarchical prediction mechanism is to estimate the QoS value accurately by leveraging hierarchical neural-regression. 
We evaluate our framework on the publicly available WS-DREAM datasets. 
The experimental results show the outperformance of our framework over the major state-of-the-art approaches.

\end{abstract}

}

\maketitle

\IEEEraisesectionheading{\section{Introduction}\label{sec:intro}}
\noindent
Services computing is becoming an emerging field of research with the focus shift towards Everything-as-a-Service (XaaS). With the proliferation of Internet-of-Things (IoT), Machine to Machine (M2M) communication and smart technologies, the number of web services is increasing day by day. The massive growth of the number of functionally equivalent services 
introduces a challenge to the service recommendation research.
Numerous ways exist in the literature to recommend the services for a specific task. For example, the service recommendation can be accomplished based on user preferences \cite{DBLP:journals/access/WuZHZZ19}, where the user specifies a set of criteria on which the recommender engine chooses the set of services. Sometimes, the performance of the web services becomes the principal concern, and the recommendation is made based on the QoS parameters of the services \cite{DBLP:conf/icsoc/ZouJNWPG18,DBLP:conf/IEEEscc/LiWSZ17}. Often, the cost of the services \cite{DBLP:conf/IEEEscc/RamacherM12} is solely responsible for recommending the services. Seldom, the feature provided by the services \cite{DBLP:conf/IEEEscc/ChattopadhyayBM16} becomes the criteria for recommendation. 

In this paper, we concentrate on the QoS parameters for recommending the services.
The set of QoS parameters has been widely adopted to differentiate the functionally equivalent services in terms of their quality. Therefore, the QoS parameter often plays a crucial role in service recommendation. One of the fundamental challenge with QoS-based service recommendation is to obtain the exact QoS value of service before its execution. In general, the service provider supplies the QoS values of a service during its deployment. However, the QoS values of service often fluctuate depending on various factors,
such as users, time, locations, etc.
Therefore, determining the QoS value of service is an essential requirement, which drives our work in this paper.

A significant amount of work has been carried out in the literature to address the QoS prediction problem. One of the principal techniques for QoS prediction is collaborative filtering \cite{DBLP:conf/icws/WuQWZY16,DBLP:journals/soca/YuH16,DBLP:conf/icws/ZhouWGP15,sun2013personalized}. 
The main idea of this technique is to predict the QoS value of a target service to be invoked by a user on the basis of QoS values of a set of similar services invoked by a set of similar users. 
The collaborative filtering is primarily of two types: \emph{memory}{-based} and \emph{model}{-based}. 
The \emph{memory}{-based} collaborative filtering \cite{zheng2011qos,sun2013personalized,wu2017collaborative,DBLP:conf/IEEEscc/LiWSZ17,DBLP:conf/icsoc/ZouJNWPG18} is further classified into two different categories: \emph{user}{-based} and \emph{service}{-based}. In the \emph{user}{-based} collaborative filtering \cite{breese1998empirical}, the main idea is to obtain a set of users similar to the target user before predicting the target QoS value, whereas, for the \emph{service}{-based} collaborative filtering \cite{sarwar2001item}, the target QoS prediction is performed by taking into account the set of similar services. To improve the QoS prediction accuracy further, both the user-based and service-based collaborative filtering are combined to predict the target QoS value \cite{DBLP:journals/access/ChenWXZZC19,DBLP:conf/IEEEscc/LiWSZ17,DBLP:conf/icws/WuQWZY16,DBLP:journals/soca/YuH16,DBLP:conf/icws/ZhouWGP15,zheng2011qos,DBLP:conf/icsoc/ZouJNWPG18}. However, the memory-based technique suffers from the sparsity problem, which is the major limitation of this technique. 

To overcome the limitation of memory-based approaches, the \emph{model}{-based} collaborative filtering has been introduced, where different models can be learned according to the characteristics of the datasets for QoS value prediction. The matrix factorization \cite{DBLP:journals/tsc/ZhengMLK13,DBLP:conf/soca/XuYL13,DBLP:conf/icws/XiongWLLH18} is one of the popular model-based QoS prediction techniques, where the QoS prediction is done based on a decomposition of a user-service matrix into low-rank matrices and followed by its reconstruction. As an improvement of matrix factorization, the regularization has been introduced further \cite{DBLP:conf/IEEEscc/QiHSGL15,DBLP:conf/wocc/LuoZXZ14,lo2012extended}. Regression is another approach for model-based QoS prediction \cite{DBLP:journals/access/AlamAZKI19,DBLP:journals/access/BissingKCSR19}. A few variations of regression models \cite{DBLP:conf/IEEEscc/ZhangSL0L16,DBLP:conf/colcom/Chen16a} have also been proposed in this context to improve prediction accuracy. 
Another group of studies combined both memory-based and model-based techniques to obtain better prediction accuracy. For example, collaborative filtering can be combined with neural regression \cite{DBLP:conf/icsoc/ChattopadhyayB19} to predict the target QoS value. However, the prediction accuracy is still not up to the mark, which we address in 
this paper. 

Here,  
we propose a novel framework CAHPHF for QoS prediction. The crux of our proposal is to incorporate the contextual information of users and services while predicting the QoS value of a target service to be invoked by a target user. Sometimes, the contextual information carries some additional knowledge about users or services, which may be essential to improve the prediction accuracy.
This fact motivates us to undertake this proposal. 

Our proposed CAHPHF is a combination of memory-based and model-based techniques, and comprises two key phases: \emph{hybrid filtering} followed by \emph{hierarchical prediction}. In the hybrid filtering phase, we combine user-based and service-based modules by leveraging the contextual information of users or services. In the hierarchical prediction phase, we first handle the sparsity problem by filling up the matrix using collaborative filtering/matrix factorization. We then employ a hierarchical neural-regression module to predict the target QoS parameter. 
The hierarchical neural-regression module comprises two 
layers: while the \emph{first} layer predicts the target QoS value, the \emph{second} layer attempts to increase the prediction accuracy. 

We now briefly mention the major contributions below.

{{\bf{\em{(i)}}} 
We present a new framework (CAHPHF) to predict the QoS value of a service to be invoked by a user while improving the prediction accuracy as compared to the state-of-the-art approaches.

{{\bf{\em{(ii)}}} Our CAHPHF, on the one hand, takes advantage of the \emph{memory}-based approaches by adopting filtering. On the other hand, it leverages \emph{model}-based approaches by introducing a hierarchical prediction mechanism.

{{\bf{\em{(iii)}}} 
We propose hybrid filtering, which is a combination of \emph{user}-intensive and \emph{service}-intensive filtering modules. 
In the \emph{user}-intensive filtering module, the user is more influential than the services, whereas, in the \emph{service}-intensive filtering module, the service has more impact than the user. 
Each of the user and service-intensive modules is again a combination of user-based and service-based filtering. Additionally, in this user/service-based filtering, we consolidate contextual information with the similarity information of the users/services  intending to increase the accuracy. 

{{\bf{\em{(iv)}}} A hierarchical prediction mechanism for QoS prediction is also proposed here. In this mechanism, we first focus on predicting the QoS value of the target service to be invoked by the target user. We then concentrate on reducing the error obtained by our framework. 

{{\bf{\em{(v)}}} 
We performed an extensive empirical study on the publicly available datasets WS-DREAM  \cite{zheng2014investigating,DBLP:conf/issre/ZhangZL11}. 
We also analyzed the impact of context-sensitivity of users/services on prediction. 

The rest of the paper is organized as follows. 
Section \ref{sec:preliminaries} presents the overview of the problem with its formulation. Section \ref{sec:method} discusses the proposed framework in detail. Section \ref{sec:result} analyzes the experimental results.
Section \ref{related} presents literature review. Finally, Section \ref{sec:conclusion} concludes this paper.

\section{Overview and Problem Formulation}\label{sec:preliminaries}
\noindent
In this section, we formalize our problem. 
We begin with defining the notion of the QoS invocation log below.

\begin{definition}{[{\bf{QoS Invocation Log:}}]} 
The QoS invocation log ${\cal{Q}}$ is a 2-dimensional matrix, where each entry ${\cal{Q}}[i, j] \in \mathbb{R}^+$ of the matrix represents the value of the 
{QoS parameter $q$} of the service $s_j$ when invoked by the user $u_i$. 
 \hfill$\blacksquare$
\end{definition}

\noindent
It may be noted that most of the time, the QoS invocation log is a sparse matrix. If a user $u_i$ invoked a service $s_j$ in the past, 
the corresponding QoS value is recorded in the QoS invocation log and represented by the entry of ${\cal{Q}}[i, j]$. However, if a user $u_i$ never invoked a service $s_j$ in the past, the corresponding entry in QoS invocation log is represented by $0$.
In other words, ${\cal{Q}}[i, j] = 0$ implies the user $u_i$ has never invoked the service $s_j$. 

\begin{example}\label{example:qos}
Consider Table \ref{tab:qos} representing the QoS invocation log ${\cal{Q}}$ for a set of 10 users ${\cal{U}}$ and a set of 10 services ${\cal{S}}$. Here, we consider the response time as the QoS parameter. 
\begin{table}[!t]\makegapedcells
\tiny
\caption{Example of QoS invocation log}
\centering
\begin{tabular}{|c|c|c|c|c|c|c|c|c|c|c|}
\cline{2-11}
 \multicolumn{1}{c|}{}& $s_1$ & $s_2$ & $s_3$ & $s_4$ & $s_5$ & $s_6$ & $s_7$ & $s_8$ & $s_9$ & $s_{10}$\\ \hline
$u_1$ &  5.98 & 0.22 & 0.23 & 0 \textcolor{blue}{$\ast$} & 0.22 & 0.52 & 0.45 & 0.56 & 0.38 & 0 \\ 
$u_2$ &  2.13 & 0.26 & 0.27 & 0.25 & 0.25 & 0 & 0.65 & 0.64 & 0.43 & 0.72 \\ 
$u_3$ &  0.85 & 0 & 0.37 & 0.35 & 0.35 & 0.11 & 0.64 & 0 & 0.64 & 1.21 \\ 
$u_4$ &  0.69 & 0.22 & 0.23 & 0.22 & 0 & 0.34 & 0.76 & 0 & 0.37 & 0.55 \\ 
$u_5$ &  0.86 & 0 & 0.23 & 0.22 & 0.22 & 0.36 & 0.83 & 0.86 & 0.37 & 0.61 \\ 
$u_6$ &  1.83 & 0.25 & 0 & 0.26 & 0.23 & 0 & 0.89 & 0.92 & 0.42 & 0.86 \\ 
$u_7$ &  0.81 & 0.24 & 0.25 & 0.23 & 0.23 & 0.25 & 0 & 0.91 & 0.43 & 0 \\ 
$u_8$ &  0 & 0.24 & 0.25 & 0 & 0.26 & 0.33 & 0.59 & 0 & 0.42 & 1.85 \\ 
$u_9$ &  2.05 & 0.21 & 0 & 0.20 & 0.2 & 0.43 & 0.45 & 0.71 & 06 & 0.64 \\ 
$u_{10}$ &  0.86 & 0 & 0.22 & 0.2 & 0.19 & 0.38 & 0.59 & 0.62 & 0 & 0.49 \\ 
\hline
\end{tabular}\label{tab:qos}
\end{table}
\noindent 
${\cal{Q}}[i, j]$ represents the value of the response time (in millisecond) of $s_j \in {\cal{S}}$ invoked by $u_i \in {\cal{U}}$. 
For example, ${\cal{Q}}[1, 1] = 5.98$ represents the value of the response time of $s_1$ invoked by $u_1$. 
 \hfill $\blacksquare$ 
\end{example}

\subsection{Problem Formulation}\label{sec:problem}
\noindent
We are given:
\begin{itemize}
 \item[---] a set of users ${\cal{U}} = \{u_1, u_2, \ldots, u_n\}$
 
 \item[---] $\forall u_i \in {\cal{U}}$, a contextual information $\alpha_i = (\phi_i, \psi_i)$.
 In this paper, we consider the location information as the contextual information of a user. $\alpha_i$ is represented by a 2-tuple, where $\phi_i$ and $\psi_i$ represent the latitude and the longitude of the \emph{user} respectively.
 
 \item[---] a set of services ${\cal{S}} = \{s_1, s_2, \ldots, s_m\}$
 
 \item[---] $\forall s_i \in {\cal{S}}$, a contextual information $\beta_i = (\phi_i, \psi_i)$. Here also we consider the location information as the contextual information. Similar to $\alpha_i$, $\beta_i$ is also represented by a 2-tuple, where $\phi_i$ and $\psi_i$ represent the latitude and the longitude of the \emph{service} respectively.
 
 \item[---] $\forall u_i \in {\cal{U}}$, a set of services ${\cal{S}}_i \subseteq {\cal{S}}$ invoked by $u_i$
 
 \item[---] $\forall s_i \in {\cal{S}}$, a set of users ${\cal{U}}_i \subseteq {\cal{U}}$ invoked $s_i$
 
 \item[---] a QoS parameter $q$ of the services
 
 \item[---] a QoS invocation log ${\cal{Q}}$ 
 
 \item[---] a target user $u_{t_1}$ and a target service $s_{t_2}$
\end{itemize}

\noindent
{The objective} of this work is to estimate accurately the QoS value (say, \emph{target QoS value}) of $s_{t_2}$ when it is to be invoked by $u_{t_1}$. 

\begin{example}\label{example:problem}
 In Example \ref{example:qos}, the objective of the QoS prediction problem is to predict the value of the response time of $s_4$ to be invoked by $u_1$ (marked as \textcolor{blue}{$\ast$} in Table \ref{tab:qos}).
 \hfill $\blacksquare$ 
\end{example}

\begin{figure*}
    \centering
 	\includegraphics[width=0.7\linewidth]{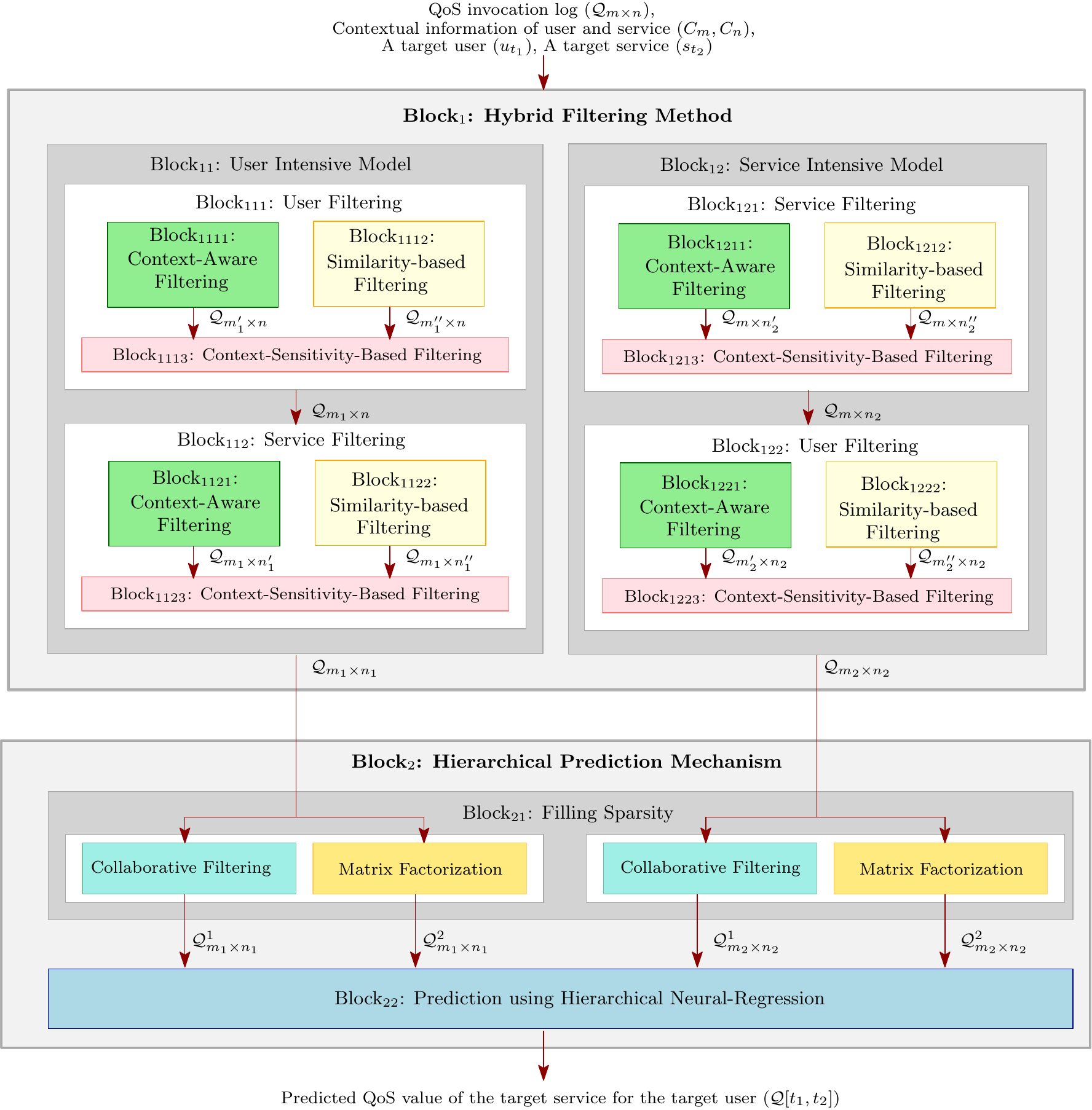}
 	\caption{Our solution framework (CAHPHF)}
 	\label{fig:architecture}
\end{figure*}

\section{Detailed Methodology}\label{sec:method}
\noindent
In this section, we discuss our solution framework to predict the target QoS value. Fig. \ref{fig:architecture} depicts the architecture of our framework (CAHPHF). The CAHPHF comprises two key steps, a \emph{hybrid filtering} method followed by a \emph{hierarchical prediction} mechanism, 
which are elaborated in the next two subsections.

\subsection{Hybrid Filtering Method}
\noindent
This is the first phase of CAHPHF (referred to Block$_1$ of Fig. \ref{fig:architecture}). Given a target user and a target service, in this phase, we filter the set of users and the services considering different perspectives and principles. 

We first consider two different aspects of filtering: 
\begin{itemize}
 \item {\em{User-intensive filtering}}: In this case, we first filter the set of users and by leveraging refined user information, we filter the set of services. 
 \item {\em{Service-intensive filtering}}: In this case, we first filter the set of services. We further filter the set of users utilizing the filtered service information.
\end{itemize}
\noindent
{While} filtering the users/services, we use two different principles:
\begin{itemize}
 \item Filtering based on {\em{contextual information}} of the users or services; 
 \item Filtering based on {\em{similarity information}} in terms of the historical QoS record.
\end{itemize}
\noindent
{Once} we obtain the filtered users/services from the above two steps, we aggregate them by analyzing the context-sensitivity of the users/services to generate the final set of similar users/services, which is the output of the hybrid filtering.

We now illustrate each of the techniques mentioned above. We begin with discussing different filtering principles. We then elaborate on the inclusion of each of these approaches in our user-intensive and service-intensive filtering methods.

\subsubsection{{\bf{Context-Aware Filtering}}}\label{subsubsec:cf}
\noindent
As mentioned earlier in Section \ref{sec:problem}, we use the location information of the users/services as the contextual information in this paper. More specifically, we use latitude and longitude to refer to the location. Before discussing the details of context-aware filtering, we first define the contextual distance between two users/services. Here, we use the Haversine distance \cite{wang2019qos} function to refer to the contextual distance. 

\begin{definition}{{\bf{[Haversine Distance $HD (\gamma_i, \gamma_j)$]:}}} Given the location of two users/services $\gamma_i$ and $\gamma_j$, the Haversine distance is computed as:
\begin{equation}\scriptsize
 \begin{split}
  HD (\gamma_i, \gamma_j) = \quad\quad\quad\quad \quad\quad\quad \quad\quad\quad\quad\quad\quad \quad\quad\quad \quad\quad\quad\quad\quad\quad \quad\quad\quad\\
  2r \times arcsin\sqrt{sin^2\left(\frac{\phi_j - \phi_i}{2}\right) + cos(\phi_i)cos(\phi_j)sin^2\left(\frac{\psi_j - \psi_i}{2}\right)}
 \end{split}
\end{equation}

\noindent
where $\gamma_i = (\phi_i, \psi_i),~ \gamma_j = (\phi_j, \psi_j)$; either $\gamma_i, \gamma_j \in \{\alpha_1, \alpha_2, \ldots, \alpha_n\}$
or
$\gamma_i, \gamma_j \in \{\beta_1, \beta_2, \ldots, \beta_m\}$. $r$ is the radius of the earth ($\approx 6371$ km) and $arcsin$ represents inverse sine function.
\hfill$\blacksquare$
\end{definition}

\noindent
{Given} a set of users/services and a target user/service, we compute a set of users/services similar to the target user/service in terms of contextual distance. Algorithm \ref{algo:contextualCluster} presents the formal algorithm for computing the contextually similar users/services.

\begin{algorithm}\scriptsize
  \caption{ClusteringBasedOnContextualInformation}
  \begin{algorithmic}[1]
    \State {\bf{Input:}} A set of users $\cal{U}$ (services $\cal{S}$), a target user $u_{t_1}$ (target service $s_{t_2}$)
    \State {\bf{Output:}} A set of users ${\cal{U}}^c_{t_1} \subseteq {\cal{U}}$ (services ${\cal{S}}^c_{t_2} \subseteq {\cal{S}}$) similar to $u_{t_1}$ ($s_{t_2}$)
    \State ${\cal{U}}^c_{t_1} = \{u_{t_1}\}$ (or ${\cal{S}}^c_{t_2} = \{s_{t_2}\}$); $\Upsilon = NULL$;
    \For {each $u_i \in {\cal{U}}^c_{t_1}$ (or $s_i \in {\cal{S}}^c_{t_2}$) and $u_i \notin \Upsilon$ (or $s_i \notin \Upsilon$)}
      \State $\Upsilon \leftarrow \Upsilon \cup \{u_i\}$ (or $\Upsilon \leftarrow \Upsilon \cup \{s_i\}$);
      \For {each $u_j \in {\cal{U}} \setminus {\cal{U}}^c_{t_1}$ (or $s_j \in {\cal{S}} \setminus {\cal{S}}^c_{t_2}$)}
        \If {$HD (\alpha_i, \alpha_j) \le T^u_c (or~HD (\beta_i, \beta_j) \le T^s_c)$}
          \State ${\cal{U}}^c_{t_1} \leftarrow {\cal{U}}^c_{t_1} \cup \{u_j\}$ (or ${\cal{S}}^c_{t_2} \leftarrow {\cal{S}}^c_{t_2} \cup \{s_j\}$);
        \EndIf
      \EndFor
    \EndFor
    \State return ${\cal{U}}^c_{t_1}$ (or ${\cal{S}}^c_{t_2}$);
  \end{algorithmic}
  \label{algo:contextualCluster}
\end{algorithm}

{The} main idea of Algorithm \ref{algo:contextualCluster} is to generate a cluster containing all the users/services similar to the target user/service in terms of contextual distance. A tunable threshold parameter $T^u_c (\text{or}~T^s_c)$ is chosen to filter the set of users (or services) to obtain the set of contextually similar users ${\cal{U}}^c_{t_1}$ (or services ${\cal{S}}^c_{t_2}$). 
If the distance between two users $u_i \in {\cal{U}}^c_{t_1}$ and $u_j \in ({\cal{U}} \setminus {\cal{U}}^c_{t_1})$ (or two services $s_i \in {\cal{S}}^c_{t_2}$ and $s_j \in ({\cal{S}} \setminus {\cal{S}}^c_{t_2})$) is less than $T^u_c (\text{or}~T^s_c)$, 
we add $u_j$ (or $s_j$) to ${\cal{U}}^c_{t_1}$ (or ${\cal{S}}^c_{t_2}$). 
We start with the target user (or service) and add it to ${\cal{U}}^c_{t_1}$ (or ${\cal{S}}^c_{t_2}$). Once a new user $u_i$ (or service $s_i$) is added to ${\cal{U}}^c_{t_1}$ (or ${\cal{S}}^c_{t_2}$), we follow the same procedure for $u_i$ (or $s_i$) as well. The algorithm terminates when there is no new user/service to be added. It may be noted that once the algorithm terminates, ${\cal{U}}^c_{t_1}$ (or ${\cal{S}}^c_{t_2}$) contains the set of users (or services) similar to the target user (or service), either directly or transitively. Here, we consider the transitive similarity, as this is required in the hierarchical prediction module to fill up the sparse matrix. 
We now consider a lemma with its proof as below.

\begin{lemma}\label{lemma:ca}
 $\forall u_i \in {\cal{U}}^c_{t_1}$, ${\cal{U}}^c_i = {\cal{U}}^c_{t_1}$, where ${\cal{U}}^c_i$ is the set of users similar to $u_i$ in terms of contextual distance.
 \hfill$\blacksquare$
\end{lemma}

\begin{proof}
 Consider $u_i \in {\cal{U}}^c_{t_1}$. To prove ${\cal{U}}^c_i = {\cal{U}}^c_{t_1}$, we need to show $\forall u_j \in {\cal{U}}^c_i$, $u_j \in {\cal{U}}^c_{t_1}$ and $\forall u_j \in {\cal{U}}^c_{t_1}$, $u_j \in {\cal{U}}^c_i$. Now consider a user $u_j (\ne u_i) \in {\cal{U}}^c_i$.  
 This implies $u_j$ is either directly or transitively similar to $u_i$ in terms of contextual distance. Again, $u_i \in {\cal{U}}^c_{t_1}$, and Haversine distance is commutative (i.e., $HD(\alpha_i, \alpha_{t_1}) = HD(\alpha_{t_1}, \alpha_i)$) imply $u_j$ is similar to $u_t$ (i.e., $u_j$ is similar to $u_i$ and $u_i$ is similar to $u_{t_1}$, thereby $u_j$ is similar to $u_{t_1}$). This in turn implies that $\forall u_j \in {\cal{U}}^c_i$, $u_j \in {\cal{U}}^c_{t_1}$. 
 
 Now consider a user $u_j (\ne u_i) \in {\cal{U}}^c_{t_1}$.  
 This implies $u_j$ is either directly or transitively similar to $u_{t_1}$ in terms of contextual distance according to Algorithm \ref{algo:contextualCluster}. $u_i$ is also similar to $u_{t_1}$, since, $u_i \in {\cal{U}}^c_{t_1}$. From the transitive property of Algorithm \ref{algo:contextualCluster} and the commutative property of Haversine distance, we can conclude $u_j$ is also similar to $u_i$ (i.e., $u_j$ is similar to $u_{t_1}$ and $u_i$ is similar to $u_{t_1}$, thereby, $u_j$ is similar to $u_i$). This implies, $\forall u_j \in {\cal{U}}^c_{t_1}$, $u_j \in {\cal{U}}^c_i$. 
 
 $\forall u_j \in {\cal{U}}^c_i$, $u_j \in {\cal{U}}^c_{t_1}$ implies  ${\cal{U}}^c_i \subseteq {\cal{U}}^c_{t_1}$.
 Again, $\forall u_j \in {\cal{U}}^c_{t_1}$, $u_j \in {\cal{U}}^c_i$ implies ${\cal{U}}^c_{t_1} \subseteq {\cal{U}}^c_i$. Both of them together imply ${\cal{U}}^c_i = {\cal{U}}^c_{t_1}$. Hence, the above lemma is true.
\end{proof}

\noindent
{The} same lemma is applicable for context-aware service filtering too, i.e., $\forall s_i \in {\cal{S}}^c_{t_2}$, ${\cal{S}}^c_i = {\cal{S}}^c_{t_2}$, where ${\cal{S}}^c_i$ is the set of services similar to $s_i$ in terms of contextual distance.


\subsubsection{{\bf{Similarity-based Filtering}}}
\label{subsubsec:sf}
\noindent
The objective of this module is to obtain a set of users/services correlated to the target user/service in terms of their past QoS invocation histories. In this paper, we use cosine similarity function \cite{zhu2015privacy} to compute the correlation between two users/services as defined below.

\begin{definition}{\bf{[Cosine Similarity Measure]:}} 
Given 2 users $u_i$ and $u_j$ (or 2 services $s_i$ and $s_j$) and their past QoS records from the QoS invocation log, the cosine similarity between the two users (or services) is defined in Equation \ref{eq:userCosine} (or Equation \ref{eq:serviceCosine}).
\begin{equation}
\label{eq:userCosine}
\scriptsize
 CSM (u_i, u_j) = \frac{\sum \limits_{s_k \in {\cal{S}}_i \cap {\cal{S}}_j} \left({\cal{Q}}[i, k] \times {\cal{Q}}[j, k]\right)}
 {\sqrt{\sum \limits_{s_k \in {\cal{S}}_i} ({\cal{Q}}[i, k])^2} \times {\sqrt{\sum \limits_{s_k \in {\cal{S}}_j} ({\cal{Q}}[j, k])^2}}}
\end{equation}
\begin{equation}
\label{eq:serviceCosine}
\scriptsize
 CSM (s_i, s_j) = \frac{\sum \limits_{u_k \in {\cal{U}}_i \cap {\cal{U}}_j} \left({\cal{Q}}[k, i] \times {\cal{Q}}[k, j]\right)}
 {\sqrt{\sum \limits_{u_k \in {\cal{U}}_i} ({\cal{Q}}[k, i])^2} \times {\sqrt{\sum \limits_{u_k \in {\cal{U}}_j} ({\cal{Q}}[k, j])^2}}}
\end{equation}
\end{definition}

\noindent
{Algorithm} \ref{algo:similarityBasedCluster} demonstrates the formal algorithm for generating the set of users/services correlated to the target user/service. Algorithm \ref{algo:similarityBasedCluster} is similar to Algorithm \ref{algo:contextualCluster}. The only difference is that instead of using Haversine distance function, cosine similarity measure is employed here. Similar to Algorithm \ref{algo:contextualCluster}, here also we use an external threshold parameter $T^u_s$ (or $T_s^s$) to filter the set of users/services. 
Like context-aware filtering (referred to Lemma \ref{lemma:ca}), the similarity-based filtering also has the same characteristics as stated by
Lemma \ref{lemma:csm}.

\begin{lemma}\label{lemma:csm}
 $\forall u_i \in {\cal{U}}^s_{t_1}$, ${\cal{U}}^s_i = {\cal{U}}^s_{t_1}$, where ${\cal{U}}^s_i$ is the set of users correlated to $u_i$ in terms of cosine similarity.
 \hfill$\blacksquare$
\end{lemma}

\noindent
{It} may be noted that the cosine similarity measure is also commutative. Therefore, using similar 
proof for Lemma \ref{lemma:ca},  we can prove Lemma \ref{lemma:csm} as well. The  same lemma is also applicable for similarity-based service filtering.

\begin{algorithm}[!t]\scriptsize
  \caption{ClusteringBasedOnSimilarity}
  \begin{algorithmic}[1]
    \State {\bf{Input:}} A set of users $\cal{U}$ (services $\cal{S}$), a target user $u_{t_1}$ (target service $s_{t_2}$) and QoS invocation log ${\cal{Q}}$
    \State {\bf{Output:}} A set of users ${\cal{U}}^s_{t_1} \subseteq {\cal{U}}$ (services ${\cal{S}}^s_{t_2} \subseteq {\cal{S}}$) similar to $u_{t_1}$ ($s_{t_2}$)
    \State ${\cal{U}}^s_{t_1} = \{u_{t_1}\}$ (or ${\cal{S}}^s_{t_2} = \{s_{t_2}\}$); $\Upsilon = \phi$;
    \For {each $u_j \in {\cal{U}}^s_{t_1}$ (or $s_j \in {\cal{S}}^s_{t_2}$) and $u_j \notin \Upsilon$ (or $s_j \notin \Upsilon$)}
      \State $\Upsilon \leftarrow \Upsilon \cup \{u_j\}$ (or $\Upsilon \leftarrow \Upsilon \cup \{s_j\}$);
      \For {For each $u_i \in ({\cal{U}} \setminus {\cal{U}}^s_{t_1})$ (or $s_i \in ({\cal{S}} \setminus {\cal{S}}^s_{t_2})$)}
        \If {$CS (u_i, u_j) \ge T^u_s (or~CS (s_i, s_j) \ge T^s_s)$}
          \State ${\cal{U}}^s_{t_1} \leftarrow {\cal{U}}^s_{t_1} \cup \{u_i\}$ (or ${\cal{S}}^s_{t_2} \leftarrow {\cal{S}}^s_{t_2} \cup \{s_i\}$);
        \EndIf
      \EndFor
    \EndFor
    \State return ${\cal{U}}^s_{t_1}$ (or ${\cal{S}}^s_{t_2}$);
  \end{algorithmic}
  \label{algo:similarityBasedCluster}
\end{algorithm}

\subsubsection{{\bf{Context-Sensitivity-based Filtering}}}
\label{subsubsec:fcs}
\noindent
The goal of this module is to aggregate the results obtained by both context-aware filtering and similarity-based filtering modules by analyzing the context-sensitivity. Once we have a set of contextually similar users ${\cal{U}}^c_{t_1}$ (or services ${\cal{S}}^c_{t_2}$) and a set of correlated users ${\cal{U}}^s_{t_1}$ (or services ${\cal{S}}^s_{t_2}$) in terms of their past QoS records, the next objective is to aggregate them. The aggregation is performed on the basis of the similarity between these two sets. If the set of common users (or services) between ${\cal{U}}^c_{t_1}$ and ${\cal{U}}^s_{t_1}$ (or, ${\cal{S}}^c_{t_2}$ and ${\cal{S}}^s_{t_2}$) is more than a given threshold $T^u_{cs}$ (or, $T^s_{cs}$), this implies the QoS invocation pattern of the target user (or service) is context sensitive. In this case, we proceed with the intersection set. Otherwise, we proceed with set of correlated users/services. We now mathematically define, the filtering based on context-sensitivity as below. 
\begin{equation}\label{eq:csu}
 {\cal{U}}^{cs}_{t_1} = 
 \begin{cases}
    {\cal{U}}^c_{t_1} \cap {\cal{U}}^s_{t_1} & \text{, if } {\cal{U}}^c_{t_1} \cap {\cal{U}}^s_{t_1} \ge T^u_{cs}\\
    {\cal{U}}^s_{t_1} & \text{, otherwise}
 \end{cases}
\end{equation}

\begin{equation}
 {\cal{S}}^{cs}_{t_2} = 
 \begin{cases}
    {\cal{S}}^c_{t_2} \cap {\cal{S}}^s_{t_2} & \text{, if } {\cal{S}}^c_{t_2} \cap {\cal{S}}^s_{t_2} \ge T^s_{cs}\\
    {\cal{S}}^s_{t_2} & \text{, otherwise}
 \end{cases}
\end{equation}

\noindent
We now prove the following lemma for filtering based on context-sensitivity, as follows.

\begin{lemma}\label{lemma:cs}
 $\forall u_i \in {\cal{U}}^{cs}_{t_1}$, ${\cal{U}}^{cs}_i = {\cal{U}}^{cs}_{t_1}$, where ${\cal{U}}^{cs}_i$ is the set of users similar to $u_i$ based on context-sensitivity.
 \hfill$\blacksquare$
\end{lemma}

\begin{proof}
 Consider $\forall u_i \in {\cal{U}}^{cs}_{t_1}$. Now, from Lemma \ref{lemma:ca}, we have ${\cal{U}}^{c}_i = {\cal{U}}^{c}_{t_1}$ and from Lemma \ref{lemma:csm}, we have ${\cal{U}}^{s}_i = {\cal{U}}^{s}_{t_1}$. Using the same threshold $T^u_{cs}$, as considered in Equation \ref{eq:csu}, we have ${\cal{U}}^{cs}_i = {\cal{U}}^{cs}_{t_1}$. Hence, the above lemma is valid.
\end{proof}

\noindent
{Similar} to the above two filtering techniques, Lemma \ref{lemma:cs} is also valid for service filtering based on context-sensitivity.

We now discuss the user-intensive and service-intensive filterings in the next two subsections involving the filtering modules (i.e., context-aware, similarity-based, context-sensitivity-based) discussed above. 

\subsubsection{\bf{User-Intensive Filtering}}
\noindent
Given a target service and a target user, in this step (referred to Block$_{11}$ of Fig. \ref{fig:architecture}), we first compute a set of users similar to the target user. Once we obtain the set of similar users, employing this information, we generate a set of services similar to the target service. Since in this filtering, we are leveraging the relevant user information to compute the set of similar services, this filtering is referred to \emph{user-intensive}. We now demonstrate this approach in details.

As discussed earlier, the user-intensive filtering module comprises of two basic blocks: (a) Block$_{111}$: user filtering block and (b) Block$_{112}$: service filtering block (referred to Fig. \ref{fig:architecture}). The main purpose of the {\em{user filtering}} block is to filter the set of users. Given a target user $u_{t_1}$, the context-aware filtering and the similarity-based filtering are performed independently on the set of users, as discussed in Sections \ref{subsubsec:cf} and \ref{subsubsec:sf}, to generate a set of contextually similar users ${\cal{U}}^c_{t_1}$ and the set of correlated users ${\cal{U}}^s_{t_1}$, respectively. Once both the filterings are done, on the basis of context-sensitivity, ${\cal{U}}^c_{t_1}$ and ${\cal{U}}^s_{t_1}$ are aggregated to obtain the final set of similar users ${\cal{U}}^{cs}_{t_1}$. Using the information of ${\cal{U}}^{cs}_{t_1}$, the {\em{service filtering}} is performed similarly. 
Given a target service $s_{t_2}$, context-aware filtering and similarity-based filtering are performed first on the set of services to generate a set of contextually similar services ${\cal{S}}^c_{t_2}$ and the set of correlated services ${\cal{S}}^s_{t_2}$, respectively. Finally, on the basis of context-sensitivity, ${\cal{S}}^{cs}_{t_2}$ is generated from ${\cal{S}}^c_{t_2}$ and ${\cal{S}}^s_{t_2}$. Algorithm \ref{algo:uiFiltering} presents the formal algorithm for user-intensive filtering.

\begin{algorithm}[!ht]\scriptsize
  \caption{UserIntensiveFiltering}
  \begin{algorithmic}[1]
    \State {\bf{Input:}} A set of users $\cal{U}$, A set of services $\cal{S}$, a target user $u_{t_1}$, a target service $s_{t_2}$, and the QoS invocation log ${\cal{Q}}$
    \State {\bf{Output:}} A filtered set of users ${\cal{U}}^{ui}_{t_1}$ and a filtered set of services ${\cal{S}}^{ui}_{t_2}$, and the modified QoS invocation log ${\cal{Q}}_{ui}$ 
    \State ${\cal{U}}^c_{t_1} \leftarrow $ ClusteringBasedOnContextualInformation(${\cal{U}}, u_{t_1}$); 
    \State ${\cal{U}}^s_{t_1} \leftarrow $ ClusteringBasedOnSimilarity(${\cal{U}}, u_{t_1}, {\cal{Q}}$); 
    \State ${\cal{U}}^{cs}_{t_1} \leftarrow$ FilteringBasedOnContextSensitivity(${\cal{U}}^c_{t_1}, {\cal{U}}^s_{t_1}$);
    \Comment\textcolor{darkgray}{description of this algorithm is discussed in Section \ref{subsubsec:fcs}}
    \State ${\cal{Q}}'_{ui} \leftarrow$ The set of rows from ${\cal{Q}}$ corresponding to the users in ${\cal{U}}^{cs}_{t_1}$;
    \State ${\cal{S}}^c_{t_2} \leftarrow $ ClusteringBasedOnContextualInformation(${\cal{S}}, s_{t_2}$); 
    \State ${\cal{S}}^s_{t_2} \leftarrow $ ClusteringBasedOnSimilarity(${\cal{S}}, s_{t_2}, {\cal{Q}}'_{ui}$); 
    \State ${\cal{S}}^{cs}_{t_2} \leftarrow$ FilteringBasedOnContextSensitivity(${\cal{S}}^c_{t_2}, {\cal{S}}^s_{t_2}$);
    \State ${\cal{Q}}_{ui} \leftarrow$ The set of columns from ${\cal{Q}}'_{ui}$ corresponding to the services in ${\cal{S}}^{cs}_{t_2}$;
    
    \State ${\cal{U}}^{ui}_{t_1} \leftarrow {\cal{U}}^{cs}_{t_1}$; ${\cal{S}}^{ui}_{t_2} \leftarrow {\cal{S}}^{cs}_{t_2}$;
    
    \State return ${\cal{U}}^{ui}_{t_1}$, ${\cal{S}}^{ui}_{t_2}$, ${\cal{Q}}_{ui}$;
  \end{algorithmic}
  \label{algo:uiFiltering}
\end{algorithm}

\subsubsection{\bf{Service-Intensive Filtering}}
\noindent
Given a target service and a target user, here (referred to Block$_{12}$ of Fig. \ref{fig:architecture}), initially, we compute a set of services similar to the target service. We then leverage related service information to obtain a set of users similar to the target user. Since in this filtering, we are utilizing relevant service information to compute the set of similar users, this filtering is termed as \emph{service-intensive}. Algorithm \ref{algo:siFiltering} presents the formal algorithm for service-intensive filtering. 
Algorithm \ref{algo:siFiltering} is similar to Algorithm \ref{algo:uiFiltering}, where the only difference is in the sequence of the user and service-intensive filterings. The set of services, in this case, is filtered before the set of users.


\begin{algorithm}[!htb]\scriptsize
  \caption{ServiceIntensiveFiltering}
  \begin{algorithmic}[1]
    \State {\bf{Input:}} A set of users $\cal{U}$, A set of services $\cal{S}$, a target user $u_{t_1}$, a target service $s_{t_2}$, and the QoS invocation log ${\cal{Q}}$
    \State {\bf{Output:}} A filtered set of users ${\cal{U}}^{si}_{t_1}$ and a filtered set of services ${\cal{S}}^{si}_{t_2}$, and the modified QoS invocation log ${\cal{Q}}_{si}$
    
    \State ${\cal{S}}^c_{t_2} \leftarrow $ ClusteringBasedOnContextualInformation(${\cal{S}}, s_{t_2}$); 
    \State ${\cal{S}}^s_{t_2} \leftarrow $ ClusteringBasedOnSimilarity(${\cal{S}}, s_{t_2}, {\cal{Q}}$); 
    \State ${\cal{S}}^{cs}_{t_2} \leftarrow$ FilteringBasedOnContextSensitivity(${\cal{S}}^c_{t_2}, {\cal{S}}^s_{t_2}$);
    \Comment\textcolor{darkgray}{description of this algorithm is discussed in Section \ref{subsubsec:fcs}}
    \State ${\cal{Q}}'_{si} \leftarrow$ The set of columns from ${\cal{Q}}$ corresponding to the services in ${\cal{S}}^{cs}_{t_2}$;
    
    \State ${\cal{U}}^c_{t_1} \leftarrow $ ClusteringBasedOnContextualInformation(${\cal{U}}, u_{t_1}$); 
    \State ${\cal{U}}^s_{t_1} \leftarrow $ ClusteringBasedOnSimilarity(${\cal{U}}, u_{t_1}, {\cal{Q}}'_{si}$); 
    \State ${\cal{U}}^{cs}_{t_1} \leftarrow$ FilteringBasedOnContextSensitivity(${\cal{U}}^c_{t_1}, {\cal{U}}^s_{t_1}$);
    \State ${\cal{Q}}_{si} \leftarrow$ The set of rows from ${\cal{Q}}'_{si}$ corresponding to the users in ${\cal{U}}^{cs}_{t_1}$;
       
    \State ${\cal{U}}^{si}_{t_1} \leftarrow {\cal{U}}^{cs}_{t_1}$; ${\cal{S}}^{si}_{t_2} \leftarrow {\cal{S}}^{cs}_{t_2}$;
    
    \State return ${\cal{U}}^{si}_{t_1}$, ${\cal{S}}^{si}_{t_2}$, ${\cal{Q}}_{si}$;
  \end{algorithmic}
  \label{algo:siFiltering}
\end{algorithm}

\noindent
{Once} we have the set of filtered users and services, we employ this information to predict the target QoS value. 
In the next section, we discuss our prediction mechanism.

\subsection{Hierarchical Prediction Mechanism}
\noindent
This is the second phase of our framework (referred to Block$_2$ of Fig. \ref{fig:architecture}). The main focus of this module is to predict the target QoS value. We employ a hierarchical neural network-based regression model to predict the QoS value. It may be noted that the filtered QoS invocation log, which is used by the neural regression module to predict the QoS value, is mostly a sparse matrix. Therefore, in the hierarchical prediction module, our first task is to fill up the absent QoS values in the sparse matrix before feeding it to the hierarchical neural network.

\subsubsection{Filling Sparsity}\label{subsubsec:fs}
\noindent
The goal of this module is to fill up the sparse invocation log matrix. In this paper, we employ collaborative filtering \cite{DBLP:conf/icsoc/ZouJNWPG18} and matrix factorization \cite{DBLP:conf/icws/AminCG12} to fill up the sparse matrix. We now briefly discuss these two approaches below.

\emph{\ref{subsubsec:fs}.1}~~{\bf{\em{Collaborative Filtering (CF)}}}: Collaborative filtering is one of the main approaches for QoS prediction. 
We first discuss the collaborative filtering method to fill up the filtered QoS invocation log matrix (${\cal{Q}}_{ui}$) obtained after user-intensive filtering. For every ${\cal{Q}}_{ui}[i, j] = 0$, we perform collaborative filtering. 
It may be noted that we do not need to filter the set of users/services anymore, since we have already performed filtering on the set of users and services in the hybrid filtering module. Moreover, it is clear from Lemma \ref{lemma:cs}, if we wish to filter the set of users/services again using our hybrid filtering module, the number of users/services will not reduce any further. Therefore, in collaborative filtering, we only need to predict the QoS value, which is done using the following two steps: (a) average QoS computation and (b) calculation of deviation. In the average QoS calculation, we take the weighted column average corresponding to service $s_j$, where the cosine similarity measures between $u_i$ and every other user in ${\cal{U}}^{ui}_t$ are used as the weights. Equation \ref{eq:ucfavg} shows the average QoS calculation for ${\cal{Q}}_{ui}[i, j]$, as denoted by $q_{ij}^{avg}$.
\begin{equation}\label{eq:ucfavg}\scriptsize
 q_{ij}^{avg} = \frac{\sum \limits_{u_k \in {\cal{U}}^{ui}_t} \Bigl(CSM (u_i, u_k) \times {\cal{Q}}_{ui}[k, j]\Bigr)}{\sum \limits_{u_k \in {\cal{U}}^{ui}_t} CSM(u_i, u_k)}
\end{equation}

\noindent
After having the average QoS value, we compute the deviation for accurate prediction. Here, we compute the deviation of the predicted value obtained by Equation \ref{eq:ucfavg} from the actual value across all services in ${\cal{S}}^{ui}_t$ and invoked by $u_i$. The weighted average of the deviation, considering the cosine similarity measures between $s_j$ and every other service in ${\cal{S}}^{ui}_t$ as the weights, is adjusted to $q_{ij}^{avg}$ to obtain the predicted value of ${\cal{Q}}_{ui}[i, j]$, as denoted by $q_{ij}^{pred}$. Equation \ref{eq:ucfdev} shows the computation for deviation.
\begin{equation}\label{eq:ucfdev}\scriptsize
 q_{ij}^{pred} = 
 \begin{cases} 
 & \left(q_{ij}^{avg} +
 \frac{\sum \limits_{\substack{s_k \in {\cal{S}}^{ui}_t \&\\ {\cal{Q}}_{ui}[i, k] \ne 0}} \left(CSM (s_j, s_k) \times |q_{ik}^{avg} - {\cal{Q}}_{ui}[i, k]|\right)}{\sum \limits_{\substack{s_k \in {\cal{S}}^{ui}_t \&\\ {\cal{Q}}_{ui}[i, k] \ne 0}} CSM(s_j, s_k)}\right)\\&\\
 & \quad \quad \quad \quad \quad \quad \quad \quad \quad \quad \quad \quad \text{, if } (q_{ik}^{avg} - {\cal{Q}}_{ui}[i, k]) \le 0\\
 &\\
 & \left(q_{ij}^{avg} -
 \frac{\sum \limits_{\substack{s_k \in {\cal{S}}^{ui}_t \&\\ {\cal{Q}}_{ui}[i, k] \ne 0}} \left(CSM (s_j, s_k) \times |q_{ik}^{avg} - {\cal{Q}}_{ui}[i, k]|\right)}{\sum \limits_{\substack{s_k \in {\cal{S}}^{ui}_t \&\\ {\cal{Q}}_{ui}[i, k] \ne 0}} CSM(s_j, s_k)}\right) \\&\\
 & \quad \quad \quad \quad \quad \quad \quad \quad \quad \quad \quad \quad \quad \quad \quad \quad \quad \quad \quad \quad \text{, otherwise} 
 \end{cases}
\end{equation}

\noindent
{Finally}, after employing the collaborative filtering to fill up every zero entry of ${\cal{Q}}_{ui}$, we have the modified QoS invocation log ${\cal{Q}}_{ui}^{cf}$, as defined below:
\begin{equation}\scriptsize
 {\cal{Q}}_{ui}^{cf}[i, j] = 
 \begin{cases}
    {\cal{Q}}_{ui}[i, j] & \text{, if } {\cal{Q}}_{ui}[i, j] \ne 0\\
    q_{ij}^{pred} & \text{, otherwise}
 \end{cases}
\end{equation}

\noindent
{The} collaborative filtering to fill up ${\cal{Q}}_{si}$ obtained after service-intensive filtering is exactly same as in user-intensive case. For every ${\cal{Q}}_{si}[i, j] = 0$, the collaborative filtering is performed. Equations \ref{eq:scfavg} and \ref{eq:scfdev} show the computations of average QoS value and deviation for each zero entry in ${\cal{Q}}_{si}$, respectively.
\begin{equation}\label{eq:scfavg}\scriptsize
 r_{ij}^{avg} = \frac{\sum \limits_{s_k \in {\cal{S}}^{si}_t} \left(CSM (s_j, s_k) \times {\cal{Q}}_{si}[i, k]\right)}{\sum \limits_{s_k \in {\cal{S}}^{si}_t} CSM(s_j, s_k)}
\end{equation}

\begin{equation}\label{eq:scfdev}\scriptsize
 r_{ij}^{pred} = 
 \begin{cases}
 & \left(r_{ij}^{avg} +
 \frac{\sum \limits_{\substack{u_k \in {\cal{U}}^{si}_t \&\\ {\cal{Q}}_{si}[k, j] \ne 0}} \left(CSM (u_i, u_k) \times |r_{kj}^{avg} - {\cal{Q}}_{si}[k, j]|\right)}{\sum \limits_{\substack{u_k \in {\cal{U}}^{si}_t \&\\ {\cal{Q}}_{si}[k, j] \ne 0}} CSM(u_i, u_k)}\right) \\
 &\\
 & \quad \quad \quad \quad \quad \quad \quad \quad \quad \quad \quad \quad \text{, if } (r_{kj}^{avg} - {\cal{Q}}_{si}[k, j]) \le 0\\
 &\\
 & \left(r_{ij}^{avg} -
 \frac{\sum \limits_{\substack{u_k \in {\cal{U}}^{si}_t \&\\ {\cal{Q}}_{si}[k, j] \ne 0}} \left(CSM (u_i, u_k) \times |r_{kj}^{avg} - {\cal{Q}}_{si}[k, j]|\right)}{\sum \limits_{\substack{u_k \in {\cal{U}}^{si}_t \&\\ {\cal{Q}}_{si}[k, j] \ne 0}} CSM(u_i, u_k)}\right) \\&\\
 & \quad \quad \quad \quad \quad \quad \quad \quad \quad \quad \quad \quad \quad \quad \quad \quad \quad \quad \quad \quad \text{, otherwise} 
 \end{cases}
\end{equation}

\noindent
{Finally}, after applying collaborative filtering on every zero entry of ${\cal{Q}}_{si}$, we have the modified QoS invocation log ${\cal{Q}}_{si}^{cf}$, as defined below:
\begin{equation}\scriptsize
 {\cal{Q}}_{si}^{cf}[i, j] = 
 \begin{cases}
    {\cal{Q}}_{si}[i, j] & \text{, if } {\cal{Q}}_{si}[i, j] \ne 0\\
    r_{ij}^{pred} & \text{, otherwise}
 \end{cases}
\end{equation}

\emph{\ref{subsubsec:fs}.2}~~{\bf{\em{Matrix Factorization (MaF)}}}: Here, we employed the classical matrix factorization method \cite{bjorck2015numerical} to generate modified QoS invocation logs ${\cal{Q}}^{mf}_{ui}$ and ${\cal{Q}}^{mf}_{si}$ by filling up ${\cal{Q}}_{ui}$ and ${\cal{Q}}_{si}$, respectively.

{It} may be noted that after filling ${\cal{Q}}_{ui}$ and ${\cal{Q}}_{si}$, we now have 4 filled up matrices ${\cal{Q}}^{cf}_{ui}$, ${\cal{Q}}^{cf}_{si}$, ${\cal{Q}}^{mf}_{ui}$, and ${\cal{Q}}^{mf}_{si}$, which are used to predict the target QoS value using hierarchical neural regression.
In the next part, we discuss the hierarchical neural regression.

\subsubsection{Prediction using Hierarchical Neural Regression}\label{subsubsec:hnr}
\noindent
This is the final phase of our framework, where we employ hierarchical neural network-based regression to predict the QoS value (referred to Fig. \ref{fig:nnBlockDiagram}). 
The hierarchical neural regression module comprises of 2 key blocks: 
(a) Block$_{222}$: Level-1 neural regression block (NRL-1), and 
(b) Block$_{223}$: aggregator block (referred to Fig. \ref{fig:nnBlockDiagram}). In addition to these two blocks, the hierarchical neural regression module contains another block, called controller (Block$_{221}$), which mainly controls the aggregator block with the help of NRL-1. In NRL-1, 
we have 4 neural networks which are used to predict the target QoS values from  ${\cal{Q}}^{cf}_{ui}$, ${\cal{Q}}^{cf}_{si}$, ${\cal{Q}}^{mf}_{ui}$, and ${\cal{Q}}^{mf}_{si}$, respectively. 
The aggregator block, then, aggregates the 4 QoS values obtained from NRL-1 with intending to increase the QoS prediction accuracy. 
In the aggregator module, we have 2 blocks, 
a Level-2 neural regression (NRL-2) block and 
a mean absolute error (MAE)-based aggregation (MAE-Ag) block. 
We invoke one of these two blocks, which is decided by the controller block 
on the basis of certain criteria as discussed latter in this section. 
We now describe each of the blocks of hierarchical neural regression module in details. We begin with describing Block$_{222}$. 

\begin{figure}[!t]
    \centering
 	\includegraphics[width=\linewidth]{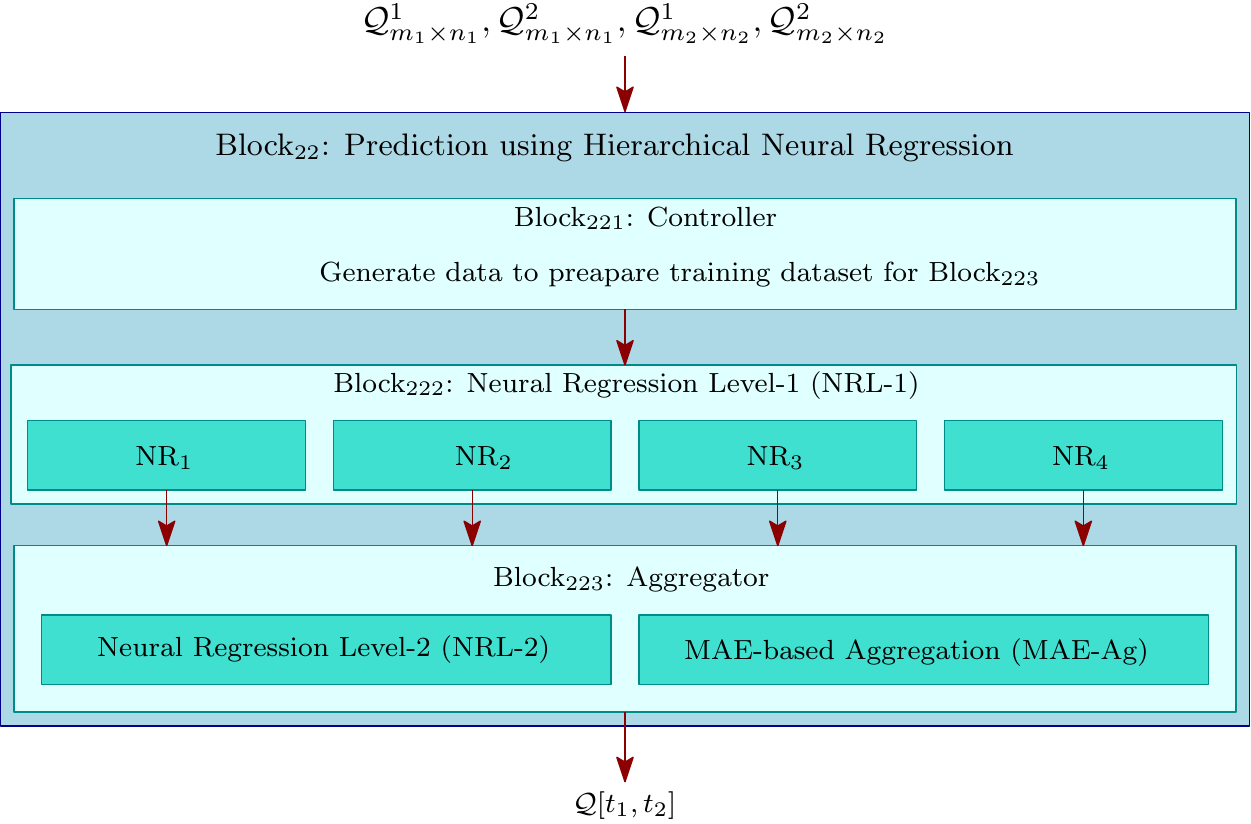}
 	\caption{Block diagram of hierarchical neural regression (inside view of Block$_{22}$ of Fig. \ref{fig:architecture})}
 	\label{fig:nnBlockDiagram}
\end{figure}

\emph{\ref{subsubsec:hnr}.1}~~{\textbf{\emph{Neural Regression Level-1 (NRL-1)}}}: This is one of the major components of the hierarchical neural regression (referred to Block$_{222}$ of Fig. \ref{fig:nnBlockDiagram}). The module consists of 4 neural regression blocks NR$_1$, NR$_2$, NR$_3$, and NR$_4$. The goal of these 4 blocks is to predict the target QoS value from 4 QoS invocation log matrices obtained from Block$_{21}$ of Fig. \ref{fig:architecture}.
Before going to the further details of NRL-1, we first describe the basic architecture of neural regression.

\textbullet~~ {{\bf{Neural Regression (NR)}}}: In this paper, we employ a neural network-based regression module to predict the target QoS value. At the top of the NR module, a linear regression \cite{lathuiliere2019comprehensive} layer is used to predict the missing QoS value. The objective of the linear regression is to generate a linear curve $h(X)$ as defined: $h(X) = \sum \limits_{i=0}^{n} \theta_i x_i$,
where $X = (x_1, x_2, \ldots, x_n)$ is a vector on the training dataset and $x_0 = 1$. The set $\theta = \{\theta_0, \theta_1, \ldots \theta_n\}$ is the set of parameters required to be optimized to predict the QoS value accurately on the test dataset. It may be noted that $\theta_0$ represents the bias.
To learn the value of each $\theta_i \in \theta$, a cost function $J(\theta)$ is used, which is defined as: $J(\theta) = \frac{1}{|TrD|} \sum \limits_{X_i \in TrD} \left(h(X_i) - y_i\right)^2$, where $y_i$ is the target value for the input vector $X_i$ and $TrD$ be the training dataset. 
To obtain the value of each $\theta_i \in \theta$, $J(\theta)$ needs to be minimized, which actually represents the mean squared error (MSE). Here, we employ a feed-forward neural network with backpropagation \cite{goodfellow2016deep} to obtain the values of $\theta$. To optimize the cost function of the neural network, we use stochastic gradient descent (SGD) with momentum \cite{goodfellow2016deep}. 


We now demonstrate the input-output of each of the 4 NR blocks below. 
Given a target user $u_{t_1}$, a target service $s_{t_2}$, a set of users ${\cal{U}}_f \in \{{\cal{U}}_{t_1}^{ui}, {\cal{U}}_{t_1}^{si}\}$, a set of services ${\cal{S}}_f \in \{{\cal{S}}_{t_2}^{ui}, {\cal{S}}_{t_2}^{si}\}$ and a QoS invocation log matrix ${\cal{Q}}_f \in \{{\cal{Q}}^{cf}_{ui}, {\cal{Q}}^{cf}_{si}, {\cal{Q}}^{mf}_{ui}, {\cal{Q}}^{mf}_{si}\}$, the training-testing setup of each NR block is discussed below.

\begin{itemize}
 \item \emph {Training Data:} For each user $u_i (\ne u_{t_1}) \in {\cal{U}}_f$, the QoS invocation record vector from ${\cal{Q}}_f$ containing the QoS values for each service in ${\cal{S}}_f$ is used as a training data. In other word, each row, except the one corresponding to $u_{t_1}$, of ${\cal{Q}}_f$ is used as a training data of the neural regression. 
 \item \emph {Input Vector:} For each training data, the QoS invocation record vector containing the QoS value of each service $s_i \in {\cal{S}}_f\setminus \{s_{t_2}\}$ is used as the input vector for neural regression. 
 \item \emph {The Target Value:} For each training data corresponding user $u_i (\ne u_{t_1})$, the QoS value of $s_{t_2}$, i.e., ${\cal{Q}}_f[i, t_2]$ is used as the target value of the neural regression.
 \item \emph {Test Input Vector:} For the user $u_{t_1}$, the QoS invocation record vector corresponding to each service $s_i \in {\cal{S}}_f\setminus \{s_{t_2}\}$ is used as the test input vector of the neural regression.
 \item \emph {Test Output:} Finally, the test output is the value of ${\cal{Q}}_f[t_1, t_2]$ to be predicted.
\end{itemize}

It may be noted that NR$_1$ and NR$_2$ work on ${\cal{Q}}^{cf}_{ui}$ and ${\cal{Q}}^{mf}_{ui}$, respectively corresponding to the user set ${\cal{U}}_{t_1}^{ui}$ and service set ${\cal{S}}_{t_2}^{ui}$, while NR$_3$ and NR$_4$ work on ${\cal{Q}}^{cf}_{si}$ and ${\cal{Q}}^{mf}_{si}$, respectively corresponding to the user set ${\cal{U}}_{t_1}^{si}$ and service set ${\cal{S}}_{t_2}^{si}$. Consider ${\cal{Q}}^{cf}_{ui}[t_1, t_2], {\cal{Q}}^{mf}_{ui}[t_1, t_2], {\cal{Q}}^{cf}_{si}[t_1, t_2],$ and ${\cal{Q}}^{mf}_{si}[t_1, t_2]$ are the 4 outputs predicted by the 4 NR modules of Block$_{222}$.

\emph{\ref{subsubsec:hnr}.2}~~{\textbf{\emph{Controller Module}}}: This is the first component (referred to Block$_{221}$ of Fig. \ref{fig:nnBlockDiagram}) of the hierarchical neural regression module. The objective of this module is to decide which component of the aggregator module to be executed and accordingly, either prepare the training dataset for NRL-2 or generate a dataset for MAE-Ag. The function of the controller is formally presented in Algorithm \ref{algo:controller}.

Given 4 QoS invocation log matrices ${\cal{Q}}^{cf}_{ui}, {\cal{Q}}^{cf}_{si}, {\cal{Q}}^{mf}_{ui},$ and ${\cal{Q}}^{mf}_{si}$ generated by Block$_{21}$, the controller module first check if it is feasible to generate the dataset to train NRL-2. The essential idea behind NRL-2 is to compare the predicted values generated by the 4 NR blocks of NRL-1 with the actual QoS value. The comparison is possible only if the 4 NR blocks of NRL-1 can predict some QoS values, which are already available to the hierarchical neural regression module. The objective of the controller module is to choose a few entries from ${\cal{Q}}_{NR} = \left({\cal{Q}}^{cf}_{ui} \cap {\cal{Q}}^{cf}_{si} \cap {\cal{Q}}^{mf}_{ui} \cap {\cal{Q}}^{mf}_{si}\right)$ to be predicted by the 4 NR blocks of NRL-1 and to be compared the predicted values with the actual QoS value by NRL-2. Therefore, to train the NRL-2, sufficient entries are required in ${\cal{Q}}_{NR}$. A given threshold $T_d$ is used to check whether ${\cal{Q}}_{NR}$ has sufficient entries to train NRL-2. If ${\cal{Q}}_{NR}$ has sufficient entries to train NRL-2, the controller chooses a random $T_d$ number of entries from ${\cal{Q}}_{NR}$ and execute 4 NR blocks of NRL-1 to obtain the predicted values. Once the predicted values from the NRL-1 are available, the controller module activates NRL-2 to obtain a more accurate predicted value. However, if ${\cal{Q}}_{NR}$ does not have sufficient entries to train the NRL-2, the controller activates the MAE-Ag. In this case, to compute the MAE value of each NR block of NRL-1, the controller chooses $T_d$ random entries from each of ${\cal{Q}}^{cf}_{ui}, {\cal{Q}}^{cf}_{si}, {\cal{Q}}^{mf}_{ui},$ and ${\cal{Q}}^{mf}_{si}$ and executes the 4 NR blocks to obtain their MAE values. Finally, the controller activates MAE-Ag. The details of both the blocks of Block$_{223}$ are discussed below. 

\begin{algorithm}[!t]\scriptsize
  \caption{Controller}
  \begin{algorithmic}[1]
    \State {\bf{Input:}} ${\cal{Q}}^{cf}_{ui}, {\cal{Q}}^{cf}_{si}, {\cal{Q}}^{mf}_{ui},$ and ${\cal{Q}}^{mf}_{si}$
    
    \If{$|{\cal{Q}}^{cf}_{ui} \cap {\cal{Q}}^{cf}_{si} \cap {\cal{Q}}^{mf}_{ui} \cap {\cal{Q}}^{mf}_{si}| \ge T_d$}
      \State $\varLambda \leftarrow NULL$;
      \State ${\cal{Q}}_{NR} \leftarrow$ Random $T_d$ entries from $\left({\cal{Q}}^{cf}_{ui} \cap {\cal{Q}}^{cf}_{si} \cap {\cal{Q}}^{mf}_{ui} \cap {\cal{Q}}^{mf}_{si}\right)$;
      \For {each $q_{ij} \in {\cal{Q}}_{NR}$}
        \State $\varphi_1 \leftarrow$ Predict the value of $q_{ij}$ using NR$_1$;
        \State $\varphi_2 \leftarrow$ Predict the value of $q_{ij}$ using NR$_2$;
        \State $\varphi_3 \leftarrow$ Predict the value of $q_{ij}$ using NR$_3$;
        \State $\varphi_4 \leftarrow$ Predict the value of $q_{ij}$ using NR$_4$;
        \State $\varLambda \leftarrow \varLambda \cup (\varphi_1, \varphi_2, \varphi_3, \varphi_4, q_{ij})$;
      \EndFor
      \State Execute NRL-2 with training data $\varLambda$;
    \Else
      \State $\varLambda_1  \leftarrow NULL; \varLambda_2  \leftarrow NULL; \varLambda_3  \leftarrow NULL; \varLambda_4 \leftarrow NULL$;
      \State ${\cal{Q}}^1_{Ag} \leftarrow$ Random $T_d$ entries from ${\cal{Q}}^{cf}_{ui}$;
      \For {each $q_{ij} \in {\cal{Q}}^1_{Ag}$}
        \State $\varphi_1 \leftarrow$ Predict the value of $q_{ij}$ using NR$_1$;
         $\varLambda_1 \leftarrow \varLambda_1 \cup (\varphi_1, q_{ij})$;
      \EndFor
      \State ${\cal{Q}}^2_{Ag} \leftarrow$ Random $T_d$ entries from ${\cal{Q}}^{mf}_{ui}$;
      \For {each $q_{ij} \in {\cal{Q}}^2_{Ag}$}
        \State $\varphi_2 \leftarrow$ Predict the value of $q_{ij}$ using NR$_2$;
        $\varLambda_2 \leftarrow \varLambda_2 \cup (\varphi_2,  q_{ij})$;
      \EndFor
      \State ${\cal{Q}}^3_{Ag} \leftarrow$ Random $T_d$ entries from ${\cal{Q}}^{mf}_{si}$;
      \For {each $q_{ij} \in {\cal{Q}}^3_{Ag}$}
        \State $\varphi_3 \leftarrow$ Predict the value of $q_{ij}$ using NR$_3$;
        $\varLambda_3 \leftarrow \varLambda_3 \cup (\varphi_3, q_{ij})$;
      \EndFor
      \State ${\cal{Q}}^4_{Ag} \leftarrow$ Random $T_d$ entries from ${\cal{Q}}^{cf}_{si}$;
      \For {each $q_{ij} \in {\cal{Q}}^4_{Ag}$}
        \State $\varphi_4 \leftarrow$ Predict the value of $q_{ij}$ using NR$_4$;
        $\varLambda_4 \leftarrow \varLambda_4 \cup (\varphi_4, q_{ij})$;
      \EndFor
      \State Execute MAE-Ag with $(\varLambda_1, \varLambda_2, \varLambda_3, \varLambda_4)$;
    \EndIf
  \end{algorithmic}
  \label{algo:controller}
\end{algorithm}

\emph{\ref{subsubsec:hnr}.3}~~{\textbf{\emph{Aggregator Module}}}: This is the second major component of the hierarchical neural regression module (referred to Block$_{223}$ of Fig. \ref{fig:nnBlockDiagram}). The objective of this module is to generate a single predicted output by combining the 4 predicted values ${\cal{Q}}^{cf}_{ui}[t_1, t_2], {\cal{Q}}^{cf}_{si}[t_1, t_2], {\cal{Q}}^{mf}_{ui}[t_1, t_2],$ and ${\cal{Q}}^{mf}_{si}[t_1, t_2]$ obtained from the previous level while minimizing the prediction error. The aggregator module comprises of two components: (a) a Level-2 neural regression module (NRL-2) and (b) a MAE-based aggregation module (MAE-Ag). We now explain each of these two modules in details.

\emph{\ref{subsubsec:hnr}.3.1}~~{\textbf{\emph{Neural Regression Level-2 (NRL-2)}}}: The objective of this module is to increase the prediction accuracy while aggregating the 4 predicted values obtained from NRL-1. 
Given a training dataset $\varLambda$, and the 4 predicted values ${\cal{Q}}^{cf}_{ui}[t_1, t_2], {\cal{Q}}^{cf}_{si}[t_1, t_2], {\cal{Q}}^{mf}_{ui}[t_1, t_2],$ and ${\cal{Q}}^{mf}_{si}[t_1, t_2]$ obtained from Block$_{222}$ as input, the training-testing setup of NRL-2 is discussed below:

\begin{itemize}
 \item {\emph{Training Data}}: Each tuple $\vartheta \in \varLambda$, generated by Algorithm \ref{algo:controller}, is used as a training data.
 \item {\emph{Input Vector}}: The input vector contains the first 4 elements of $\vartheta \in \varLambda$.
 \item {\emph{The Target Value}}: The last element of $\vartheta \in \varLambda$ is used as the target value.
 \item {\emph{Test Vector}}: The tuple $({\cal{Q}}^{cf}_{ui}[t_1, t_2], {\cal{Q}}^{cf}_{si}[t_1, t_2],$ ${\cal{Q}}^{mf}_{ui}[t_1, t_2],~ {\cal{Q}}^{mf}_{si}[t_1, t_2])$ is used as the test vector.
 \item {\em{Test Output}}: The test output is the value of ${\cal{Q}}_f[t_1, t_2]$, which is to be predicted by Level-2 neural regression.
\end{itemize}

\emph{\ref{subsubsec:hnr}.3.2}~~{\textbf{\emph{MAE-based Aggregation (MAE-Ag)}}}: The objective of this block is to compute the minimum MAE value among the 4 MAE values obtained from the 4 NR blocks of NRL-1 and accordingly chooses the NR block (say, MinNR block) generating the minimum MAE among the MAE values generated by all 4 NR blocks. Finally, MAE-Ag outputs the predicted value obtained by MinNR block, which is the final output of our framework. Algorithm \ref{algo:MAE-Ag} formally presents the MAE-Ag block.

\begin{algorithm}[!t]\scriptsize
  \caption{MAE\_BasedAggregation}
  \begin{algorithmic}[1]
    \State {\bf{Input:}} $(\varLambda_1, \varLambda_2, \varLambda_3, \varLambda_4)$
    \State {\bf{Output:}} The predicted value for ${\cal{Q}}_f[t_1, t_2]$
    
    \State Compute MAE$_1$, MAE$_2$, MAE$_3$ and MAE$_4$ from $\varLambda_1, \varLambda_2, \varLambda_3,$ and $\varLambda_4$ respectively;
    \State $MAE_{min} \leftarrow \min (MAE_1, MAE_2, MAE_3, MAE_4)$;
    \State $MinNR \leftarrow$ The NR block of NRL-1 generated $MAE_{min}$;
    \State Return the predicted value generated by $MinNR$;
  \end{algorithmic}
  \label{algo:MAE-Ag}
\end{algorithm}

In the following section, we provide rigorous experimental studies to justify the necessity of each block of our framework.

\section{Experimental Results}\label{sec:result}
\noindent
In this section, we present the experimental results with analysis. We implemented our proposed framework in MATLAB R2019b. 
All experiments were executed on the MATLAB Online\footnote{\url{https://matlab.mathworks.com}} server.

\subsection{DataSets}
\noindent
We used 2 datasets WS-DREAM-1 \cite{zheng2014investigating} and WS-DREAM-2 \cite{DBLP:conf/issre/ZhangZL11} to analyze the performance of our framework. 

(a) \emph{WS-DREAM-1}: The first dataset contains 5825 services and 339 users. The location of each user and service is also provided in the dataset. The dataset comprises 2 QoS parameters: response time (RT) and throughput (TP). For each QoS parameter, the QoS invocation log with dimension $339 \times 5825$ is also given. We used both the QoS parameters to validate our approach. 

(b) \emph{WS-DREAM-2}: The second dataset contains 4500 services, 142 users, and 64 time slices. However, this dataset does not contain any contextual information. Similar to WS-DREAM-1, this dataset also comprises 2 QoS parameters: response time and throughput. For each QoS parameter, the QoS invocation log is also provided. For each time slice, we extracted the user-service matrix with dimension $142 \times 4500$. We performed the experiment on 64 different matrices and recorded the average QoS prediction accuracy.

For experimental analysis, we divided each dataset into training, validation and testing sets. The size of the training dataset to be $x\%$ indicates $(100 - x)\%$ entries of our QoS invocation log were randomly made as 0. The remaining $(100 - x)\%$ entries are subdivided into validation and testing set into $1:2$ ratio.


We randomly chose 200 instances from the testing dataset for performance analysis. Each experiment was performed in 5 episodes for a given training dataset to find the prediction accuracy. Finally, the average result was calculated and reported in this paper.

\subsection{Comparison Metric} \noindent
We used the following two metrics to analyze our experimental results.

\begin{definition}{{\bf{[Mean Absolute Error (MAE)]}}:}
  MAE is the arithmetic mean of the absolute difference between the predicted QoS value and the actual QoS value of a service invoked by a user over the testing dataset.
  {\scriptsize{
  \begin{equation}
    MAE = \frac{1}{|TD|} \sum \limits_{q_{ij} \in TD}\left|q_{ij} - q^{pred}_{ij}\right| 
  \end{equation}}}
  \noindent
  where $q_{ij}$ is the actual QoS value, $q^{pred}_{ij}$ is the predicted QoS value, and $TD$ is the testing dataset.
\hfill$\blacksquare$ 
\end{definition}

\begin{definition}{{\bf{[Improvement $I(M_1, M_2)$]}}:}
  Given two MAE values obtained by two methods $M_1$ and $M_2$, the improvement of $M_1$ with respect to $M_2$ is defined as:
  {\scriptsize{
  \begin{equation}
    I(M_1, M_2) = \frac{MAE_2 - MAE_1}{MAE_2} \times 100 \%
  \end{equation}}}
  \noindent
  where $MAE_1$ and $MAE_2$ are the MAE values obtained by methods $M_1$ and $M_2$, respectively.
\hfill$\blacksquare$ 
\end{definition}

%

\subsection{Comparison Methods}
\noindent
We compare CAHPHF with a set of state-of-the-art methods followed by some intermediate methods to justify the performance of CAHPHF.  

\subsubsection{State-of-the-Art Methods}
\noindent
We compared the performance of CAHPHF with the following state-of-the-art approaches.

{{\bf{\em{(i)}}}~{\bf{UPCC \cite{breese1998empirical}}}: This method exercised a user-based collaborative filtering approach for QoS prediction, where only the set of users similar to the target users were generated before applying the prediction strategy.

{{\bf{\em{(ii)}}}~{\bf{IPCC \cite{sarwar2001item}}}: In this method, service-based collaborative filtering was applied for QoS prediction. The main idea in this work was to obtain a set of services similar to the target services first, followed by applying a prediction strategy.

{{\bf{\em{(iii)}}}~{\bf{WSRec \cite{zheng2011qos}}}: This method combined both the strategies as demonstrated in UPCC and IPCC.

{{\bf{\em{(iv)}}}~{\bf{NRCF \cite{sun2013personalized}}}: This method used normal recovery collaborative filtering to improve prediction accuracy.

{{\bf{\em{(v)}}}~{\bf{RACF \cite{wu2017collaborative}}}: This method also used collaborative filtering approach. Here, a ratio-based similarity metric was considered to compute the similarity, and the result was calculated by the similar users or similar services.

{{\bf{\em{(vi)}}}~{\bf{RECF \cite{DBLP:conf/icsoc/ZouJNWPG18}}}: A reinforced collaborative filtering approach was used in this work to predict the QoS value. This method integrated both user-based and service-based similarity information into a singleton collaborative filtering.

{{\bf{\em{(vii)}}}~{\bf{MF \cite{lo2012extended}}}: An extended matrix factorization-based approach was adopted in this work for prediction.

{{\bf{\em{(viii)}}}~{\bf{HDOP \cite{wang2016multi}}}: This method used multi-linear-algebra-based concepts of tensor for QoS value prediction. Tensor decomposition and reconstruction optimization algorithms were used to predict the QoS value.

{{\bf{\em{(ix)}}}~{\bf{TA \cite{DBLP:conf/icws/AminCG12}}}: This approach integrated time series ARIMA and GARCH models to predict the QoS value.

{{\bf{\em{(x)}}}~{\bf{NMF \cite{lee1999learning}}}: This approach proposed a non-negative matrix factorization method.

{{\bf{\em{(xi)}}}~{\bf{PMF \cite{mnih2008probabilistic}}}: This paper proposed a probabilistic matrix factorization method.

{{\bf{\em{(xii)}}}~{\bf{NIMF \cite{DBLP:journals/tsc/ZhengMLK13}}}: This approach proposed a user collaboration followed by a  neighborhood-integrated matrix factorization for personalized QoS value prediction.

{{\bf{\em{(xiii)}}}~{\bf{CNR \cite{DBLP:conf/icsoc/ChattopadhyayB19}}}: This method considered both user-intensive and service-intensive filtering in the filtering stage, and finally aggregated the output using the intersection between the set of users and services obtained from two different filtering methods. For prediction, this approach used neural regression to estimate the QoS value.

\subsubsection{Intermediate Methods}\label{subsubsec:inter}
\noindent
We compared the performance of CAHPHF with the following intermediate methods, which is actually the ablation study of our system.

{{\bf{\em{(i)}}}~{\bf{User-intensive Matrix Factorization (UMF)}}: In this approach, context-sensitive user-intensive filtering is performed first. Finally, matrix factorization is used on the filtered QoS invocation log matrix to predict the target QoS value.

{{\bf{\em{(ii)}}}~{\bf{Service-intensive Matrix Factorization (SMF)}}: This approach is similar to UMF. The only difference is here, instead of applying the user-intensive filtering, the service-intensive filtering is adopted. In SMF, the context-sensitive service-intensive filtering is performed first, followed by a matrix factorization method to predict the QoS value.

{{\bf{\em{(iii)}}}~{\bf{User-intensive Collaborative Filtering (UCF)}}: Here, the context-sensitive user-intensive filtering is performed first. Finally, collaborative filtering is used on the filtered QoS invocation log matrix to predict the target QoS value.

{{\bf{\em{(iv)}}}~{\bf{Service-intensive Collaborative Filtering (SCF)}}: Here, the context-sensitive service-intensive filtering is performed first. Collaborative filtering is then employed to predict the desired QoS value.

{{\bf{\em{(v)}}}~{\bf{User-intensive Neural Regression (UNR)}}: Similar to the UMF and UCF, in this approach, the context-sensitive user-intensive filtering is performed at the initial phase. However, unlike other methods, here, neural network-based regression is used on the filtered QoS invocation log matrix to predict the target QoS value.

{{\bf{\em{(vi)}}}~{\bf{Service-intensive Neural Regression (SNR)}}: This approach is similar to UNR. The only difference is, in the filtering stage of this approach, the context-sensitive service-intensive filtering is performed.

{{\bf{\em{(vii)}}}~{\bf{User-intensive Matrix factorization with Neural Regression (UMNR)}}: This approach is similar to UNR. However, in this approach, before applying neural network-based regression to predict the QoS value, matrix factorization is used to fill up the sparse matrix.

{{\bf{\em{(viii)}}}~{\bf{Service-intensive Matrix factorization with Neural Regression (SMNR)}}: This approach is similar to SNR. However, like UMNR, in this approach, before applying neural network-based regression to predict the QoS value, matrix factorization is used to fill up the sparse matrix.

{{\bf{\em{(ix)}}}~{\bf{User-intensive Collaborative Neural Regression (UCNR)}}: Unlike UMNR, here, before applying neural network-based regression to predict the QoS value, collaborative filtering is used to fill up the sparse matrix.

{{\bf{\em{(x)}}}~{\bf{Service-intensive Collaborative Neural Regression (SCNR)}}: Unlike SMNR, here, before applying neural network-based regression to predict the QoS value, collaborative filtering is used to fill up the sparse matrix like UCNR.

{{\bf{\em{(xi)}}}~{\bf{CAHPHF Without NRL-1 (CAHPHFWoNN)}}: This method is identical to the CAHPHF method except for one segment. In CAHPHFWoNN, instead of using NRL-1, collaborative filtering and matrix factorization are used to generate the training and testing dataset for NRL-2.

{{\bf{\em{(xii)}}}~{\bf{CAHPHF with minimum MAE (CAHPHF-MAE)}}: This method is also similar to the CAHPHF method except for one portion. In CAHPHF-MAE, instead of using NRL-2, here, MAE-based aggregation method is used in $Block_{223}$ of Fig. \ref{fig:nnBlockDiagram} to reduce error.

{{\bf{\em{(xiii)}}}~{\bf{User-intensive Collaborative Neural Regression Without Contextual Filtering (UCNRWoCF)}}: This method is similar to the UCNR without contextual filtering.

{{\bf{\em{(xiv)}}}~{\bf{Service-intensive Collaborative Neural Regression Without Contextual Filtering (SCNRWoCF)}}: This method is similar to the SCNR without contextual filtering.

{{\bf{\em{(xv)}}}~{\bf{User-intensive Matrix factorization with Neural Regression Without Contextual Filtering (UMNRWoCF)}}: This method is similar to the UMNR without contextual filtering.

{{\bf{\em{(xvi)}}}~{\bf{Service-intensive Matrix factorization with Neural Regression Without Contextual Filtering (SMNRWoCF)}}: This method is similar to the SMNR without contextual filtering.

{{\bf{\em{(xvii)}}}~{\bf{CAHPHF Without Contextual Filtering (CAHPHFWoCF)}}: This method is similar to the CAHPHF without contextual filtering.

\begin{table*}\makegapedcells
\scriptsize
\caption{Performance of CAHPHF and comparison with various state-of-the-art methods over WS-DREAM-1 \cite{zheng2014investigating}}
\centering
 \begin{tabular}{c|c|c|c|c|c|c|c|c|c|c|c}
 \hline
 \multirow{2}{*}{QoS} &  \multirow{2}{*}{$|TrD|$}  &    \multicolumn{9}{c|}{MAE} & \multirow{2}{*}{$I$(CAHPHF, CNR)}\\
 \cline{3-11}
  &          & {\bf{UPCC}} & {\bf{IPCC}} & {\bf{WSRec}} & {\bf{MF}} & {\bf{NRCF}} & {\bf{RACF}} & {\bf{RECF}} & {\bf{CNR}} & {\bf{CAHPHF}} & \\
 \hline
 \hline
 \multirow{3}{*}{RT} & 10\% & 0.6063 & 0.7 & 0.6394 & 0.5103 & 0.5312 & 0.4937 & 0.4332 & 0.2597 & {\bf{0.059}} & 77.28\%\\
 \cline{2-12}
 & 20\% & 0.5379 & 0.5351 & 0.5024 & 0.4981 & 0.4607 & 0.4208 & 0.3946 & 0.1711 & {\bf{0.0419}} & 75.51\%\\
 \cline{2-12}
 & 30\% & 0.5084 & 0.4783 & 0.4571 & 0.4632 & 0.4296 & 0.3997 & 0.3789 & 0.0968 & {\bf{0.0399}} & 58.78\%\\
 \hlineB{2.5}
  QoS & $|TrD|$ & {\bf{UPCC}} & {\bf{IPCC}} & {\bf{WSRec}} & {\bf{NMF}} & {\bf{PMF}} & \multicolumn{2}{c}{{\bf{NIMF}}} & \multicolumn{2}{|c|}{{\bf{CAHPHF}}} & $I$(CAHPHF, NIMF)\\
 \hline
 \hline
 \multirow{3}{*}{TP}& 5\% & 26.123 & 29.265 & 25.8755 & 25.752 & 19.9034 & \multicolumn{2}{c}{17.9297} & \multicolumn{2}{|c|}{{\bf{7.907}}} & 55.90\%\\
 \cline{2-12}
 & 10\% & 21.2695 & 27.3993 & 19.9754 & 17.8411 & 16.1755 & \multicolumn{2}{c}{16.0542} & \multicolumn{2}{|c|}{{\bf{5.98}}} & 62.75\%\\
 \cline{2-12}
 & 20\% & 17.5546 & 25.0273 & 16.0762 & 15.2516 & 14.6694 & \multicolumn{2}{c}{13.7099} & \multicolumn{2}{|c|}{{\bf{4.189}}} & 69.45\%\\
 \hline
 \multicolumn{12}{r}{$|TrD|$: Size of Training Data} \\
\end{tabular}\label{tab:compareDS1}
\end{table*}

\subsection{Configuration of CAHPHF}
\label{subsec:config}
\noindent
In this subsection, we discuss the configuration used for CAHPHF method. We set the tunable parameters empirically from our validation dataset.

\textbullet~~ {\bf{Configuration of Hybrid Filtering}}: 
Here, a set of threshold parameters, more precisely 12 threshold parameters, is required. All these threshold parameters are data-driven.

The hybrid filtering comprises of 3 different filtering techniques.
For \emph{context-aware filtering}, the threshold values were chosen as:
\begin{equation}\label{eq:tr_ca_u}\small
 T^u_c = median \left(HD(\alpha_t, \alpha_i)\right), \text{over all users } u_i \in {\cal{U}}
\end{equation}
\begin{equation}\label{eq:tr_ca_s}\small
 T^s_c = median(HD(\beta_t, \beta_i)), \text{over all services } s_i \in {\cal{S}}
\end{equation}
\noindent
where $\alpha_t$, $\beta_t$ are the contextual information of $u_t$ and $s_t$ respectively.

For \emph{similarity-based filtering}, the threshold parameters were chosen as follows.
\begin{equation}\label{eq:tr_s_u}\small
 T^u_s = \max \left(0.5 \times \max \limits_{\substack{u_i \in {\cal{U}}\\ u_i \ne u_t}} \left(CSM(u_t, u_i)\right), median \left(CSM(u_t, u_i)\right)\right)
\end{equation}
\begin{equation}\label{eq:tr_s_s}\small
 T^s_s = \max \left(0.5 \times \max \limits_{\substack{s_i \in {\cal{S}}\\ s_i \ne s_t}} \left(CSM(s_t, s_i)\right), median \left(CSM(s_t, s_i)\right)\right)
\end{equation}
\noindent
where the median was computed across ${\cal{U}}$ for user filtering and  ${\cal{S}}$ for service filtering. 
As a matter of fact, the median of the similarity values can be 0, since the cosine similarity value varies from 0 to 1. Therefore, Equations \ref{eq:tr_s_u} and \ref{eq:tr_s_s} are different from Equations \ref{eq:tr_ca_u} and \ref{eq:tr_ca_s}.

For \emph{context-sensitivity-based filtering}, we used half of the cardinality value of the set of users or services obtained after similarity-based filtering.
\begin{equation}\label{eq:tr_cs_u}\small
 T^u_{cs} = 0.5 \times |{\cal{U}}_c|
\end{equation}
\begin{equation}\label{eq:tr_cs_s}\small
 T^s_{cs} = 0.5 \times |{\cal{S}}_c|
\end{equation}

\textbullet~~ {\bf{Configuration of NRL-1}}: 
Here, 4 neural networks are used. In our experiment, each neural network consisted of 2 hidden layers comprising of 256 and 128 neurons respectively. We used the following hyper-parameters for each neural network. The learning rate was set to 0.01 with momentum 0.9. The training was performed up to 50 epochs or up to a minimum gradient of $10^{-5}$.

\textbullet~~ {\bf{Configuration of NRL-2}}: 
Only 1 neural network is used in NRL-2 for which we used the following configuration. The network comprised of only 1 hidden layer with 2 neurons. Among the hyper-parameters, the learning rate was set to 0.01 with momentum 0.9. The training was performed up to 1000 epochs or up to a minimum gradient of $10^{-5}$. The cardinality of the training dataset to train NRL-2 was taken as 200.

We provide a rigorous analysis of parameter tuning latter in this section.

\subsection{Experimental Analysis}
\noindent
In Table \ref{tab:compareDS1}, we present the performance of our CAHPHF in terms of MAE considering the response time (RT) and throughput (TP) individually as the QoS parameter over WS-DREAM-1.
Along with the performance study of CAHPHF, Table \ref{tab:compareDS1} shows a comparative study with the major state-of-the-art approaches.


In WS-DREAM-2, the contextual information of the users and services are not present. 
Therefore, instead of CAHPHF, Table \ref{tab:compareDS2} shows the comparison between CAHPHFWoCF and the major state-of-the-art approaches considering RT and TP individually over WS-DREAM-2.

\begin{table}[!t]\makegapedcells
\scriptsize
\caption{Comparison between CAHPHFWoCF with various state-of-the-art methods over WS-DREAM-2 \cite{DBLP:conf/issre/ZhangZL11}}
\centering
 \begin{tabular}{c|c|c|c|c|c|c}
 \hline
 \multirow{2}{*}{QoS} & \multirow{2}{*}{$|TrD|$}  &  \multicolumn{4}{c|}{MAE} & \multirow{2}{*}{$I$}\\
  \cline{3-6}
  &  & {{\bf{MF}}} & {{\bf{TA}}} & {{\bf{HDOP}}} & {{\bf{CAHPHFWoCF}}} & \\
 \hline
 \hline
 \multirow{3}{*}{RT}& 10\% & {0.4987} & {0.6013} & {0.3076} & {\bf{0.1187}} & 61.41\%\\
 \cline{2-7}
 & 20\% & {0.4495} & {0.5994} & {0.2276} & {\bf{0.0758}} & 66.70\%\\
 \cline{2-7}
 & 50\% & {0.4013} & {0.4877} & {0.1237} & {\bf{0.0326}} & 73.65\%\\
 \hlineB{2.5}
 
 \multirow{3}{*}{TP} & 10\% & {16.3214} & {17.2365} & {13.2578}  & {\bf{4.897}} & 63.06\%\\
 \cline{2-7}
 & 20\% & {14.1478} & {15.0994} & {10.1276}  & {\bf{4.101}} & 59.51\%\\
 \cline{2-7}
 & 50\% & {14.9013} & {14.9870} & {10.0037} & {\bf{3.561}} & 64.40\%\\
 \hline
 \multicolumn{7}{r}{$|TrD|$: Size of Training Data, $I$: $I$(CAHPHFWoCF, HDOP)}
\end{tabular}\label{tab:compareDS2}
\end{table}

From the experimental results, we have the following observations:

\begin{enumerate}[label=(\roman*)]
 \item 
 As evident from Tables \ref{tab:compareDS1} and \ref{tab:compareDS2}, our methods performed the best as compared to the major state-of-the-art approaches for both the QoS parameters RT and TP. The last columns of the above tables show the improvement of our method over the second-best method of Tables \ref{tab:compareDS1} and \ref{tab:compareDS2}.
 On average our framework achieves 65.7\% improvement as compared to the second-best state-of-the-art method of Tables \ref{tab:compareDS1} and \ref{tab:compareDS2}.
 \item It is observed from Fig.s \ref{fig:soa} (a) and (b) that MAE value decreases and thus, the prediction accuracy improves with the increase in the size of the training dataset.
\end{enumerate}

\begin{figure}[!b]
\centering
\includegraphics[width=0.495\linewidth]{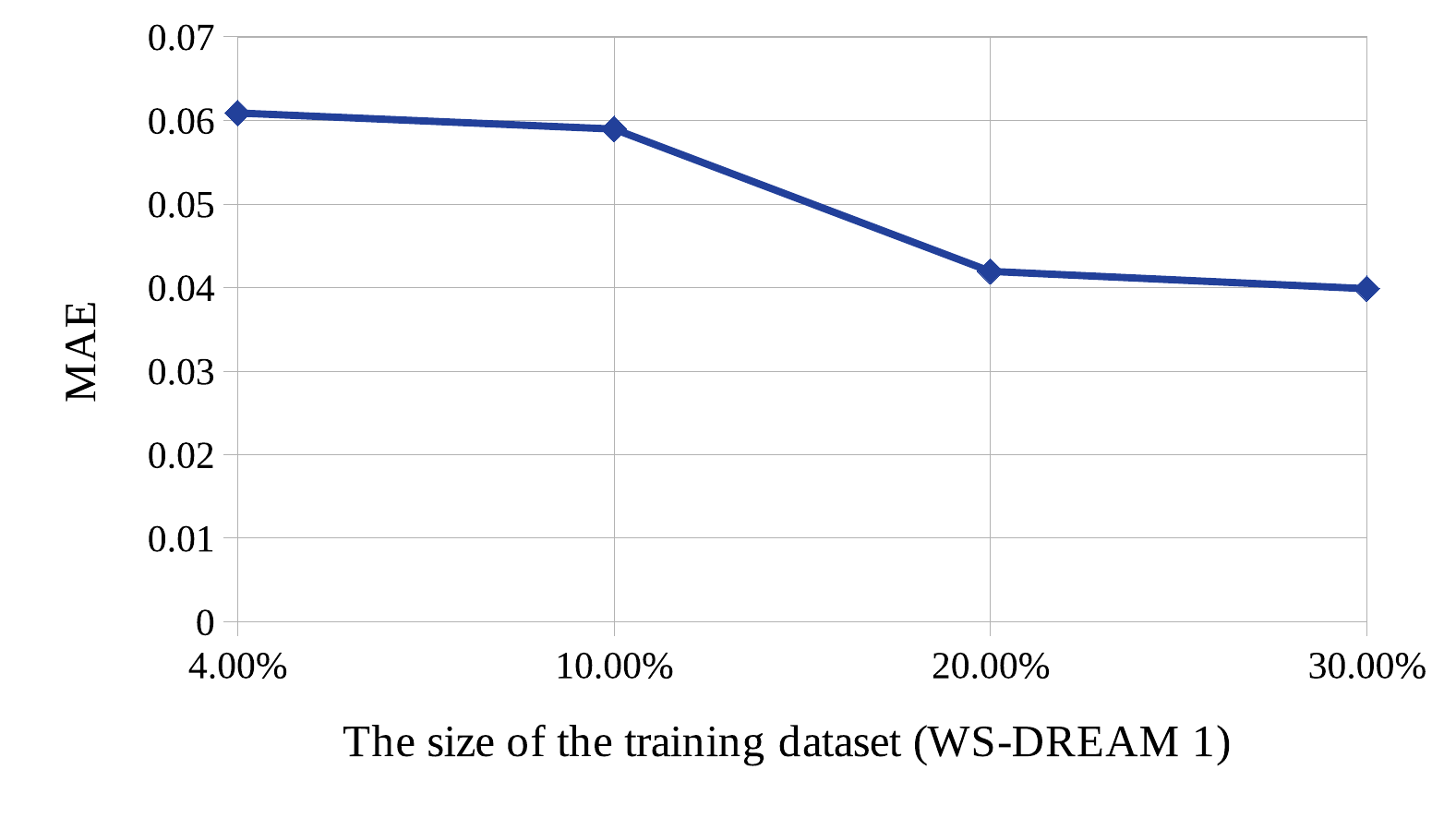}
\includegraphics[width=0.495\linewidth]{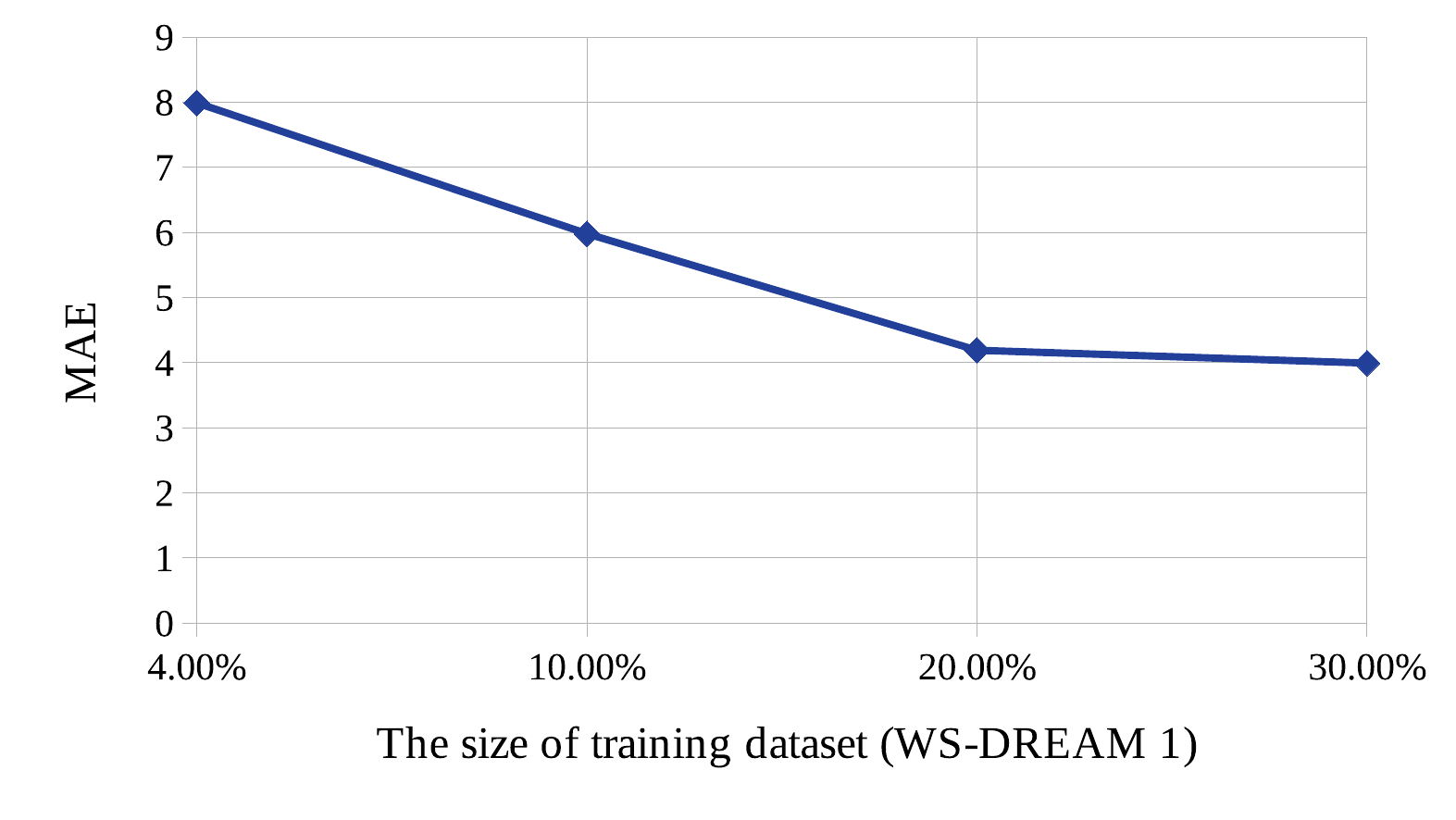}\\
(a) ~~~~~~~~~~~~~~~~~~~~~~~~~~~~~~~~~~~~~~~~~~~(b)  
\caption{Change in MAE with respect to the size of training dataset over (a) RT, (b) TP}
\label{fig:soa}
\end{figure}

\noindent
From now on we show the experimental results on response time of WS-DREAM-1, since for the rest of the datasets, similar trends are observed. 

Fig. \ref{fig:inter} shows the comparison between
CAHPHF and the intermediate methods (referred to Section \ref{subsubsec:inter}). As evident from Fig. \ref{fig:inter}, CAHPHF produced the best results compared to all the intermediate methods. The significance of each block of CAHPHF is discussed below.

\begin{enumerate}[label=(\roman*)]
 \item UNR and SNR are better than UMF, SMF, UCF, and SCF, which signifies the importance of the neural regression used in our framework.
 \item UMNR, SMNR, UCNR, and SCNR are better than UNR and SNR, which clearly establishes the requirement of our hierarchical prediction strategy, more specifically the need for filling up the sparse matrix. 
 \item CAHPHF is better than UMNR, SMNR, UCNR, and SCNR, which signifies the requirement of the hierarchical neural regression (more specifically, Block$_{223}$) module and as well as hybrid filtering.
 \item CAHPHF is also better than CAHPHFWoNN, which also signifies the requirement of the hierarchical neural regression (more specifically, Block$_{222}$) module.
 \item Finally, CAHPHF is better than CAHPHF-MAE, which clearly shows the importance of the Level-2 neural regression module.
\end{enumerate}

\begin{figure}[!t]
	\centering
	\includegraphics[width=\linewidth]{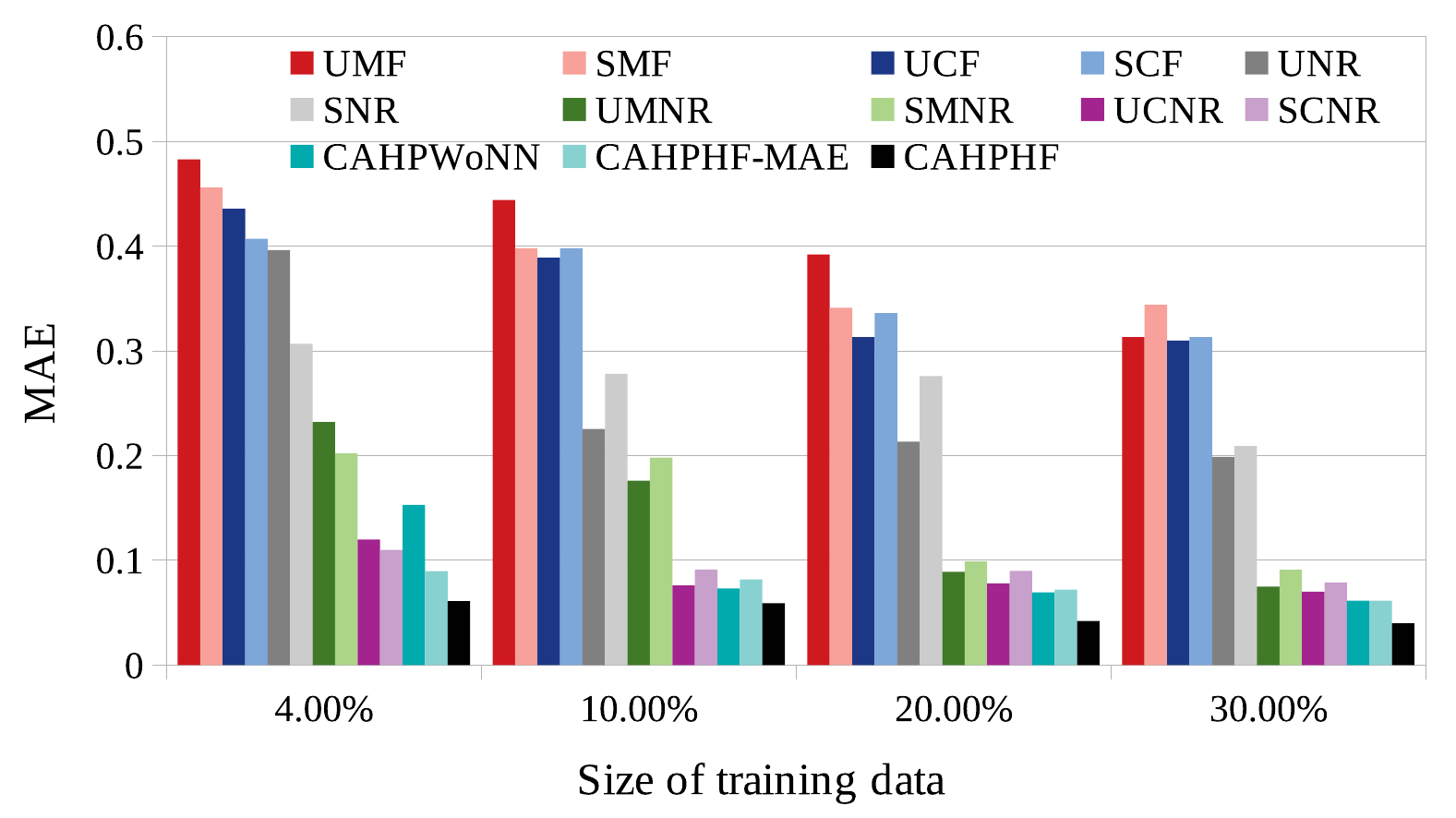}
	\caption{Comparison with intermediate methods}
	\label{fig:inter}
\end{figure} 

\begin{figure}[!hb]
    \centering
 	\includegraphics[width=\linewidth]{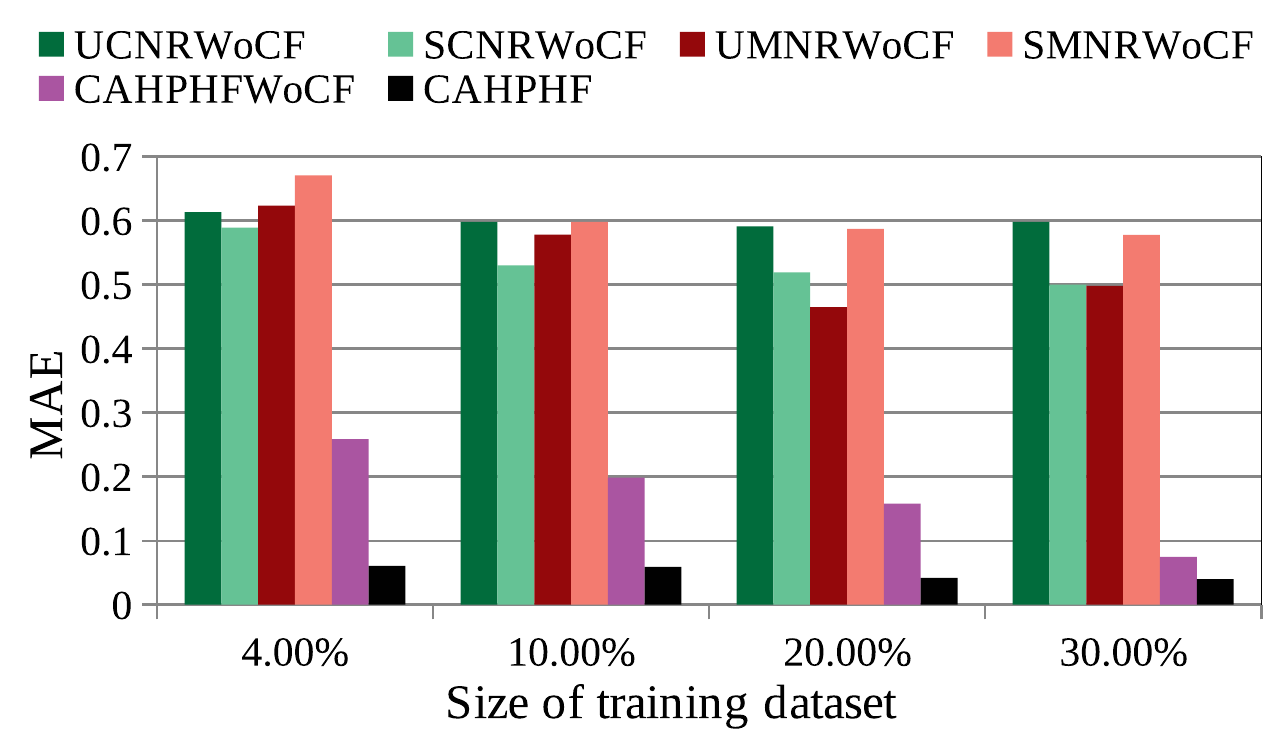}
 	\caption{Context-sensitivity analysis}
 	\label{fig:context}
\end{figure} 

\subsubsection{Context-Sensitivity Analysis}
\noindent
Fig. \ref{fig:context} shows the significance of context-sensitive filtering used in our approach. We compared our proposed method with the intermediate methods without context sensitive filtering (referred to Section \ref{subsubsec:inter}). As evident from Fig. \ref{fig:context}, CAHPHF outperformed the rest of the methods. Here, we have the following observations:

\begin{enumerate}[label=(\roman*)]
 \item CAHPHFWoCF is better than SCNRWoCF, SMNRWoCF, UCNRWoCF, and UMNRWoCF, which implies the significance of our hybrid filtering module.
 \item As CAHPHF is better than CAHPHFWoCF, it implies that the context sensitivity analysis is quite impactful in our framework.
\end{enumerate}

\subsubsection{Impact of Tunable Parameters}\label{subsubsec:itp}
\noindent
In this subsection, we show the impact of tunable parameters on the prediction accuracy of CAHPHF. In this analysis, for each experiment, we varied one parameter at a time, while keeping the rest of the parameters as constant, same as discussed in Section \ref{subsec:config}. 
We begin with discussing the impact of the threshold parameter required in hybrid filtering module on the prediction accuracy.

{\emph{\ref{subsubsec:itp}.1}~ {\textbf{\emph{Impact of Threshold Parameter:}}}
Before analyzing the impact of the threshold parameters on the prediction accuracy, we first discuss the tuning of these parameters. 
In this experiment, we chose a variable $k~(0 \le k \le 1)$ to tune the threshold parameters, as explained below. 

{{\bf{\em{(i)}}}~ In context aware filtering, the set of  users/services to be filtered was sorted in descending order based on their contextual distances from the target user/service. 
From the sorted list, the contextual distance (say, $\lambda^H_k$) between the $(\ceil{|\Gamma| \times k})^{th}$ user/service and the target user/service was chosen as threshold (i.e., $T_c^u$/$T_c^s$), where $\Gamma$ be the set of users/services to be filtered. It may be noted, by varying $k$, we obtained different values of $T_c^u$/$T_c^s$.
 
{{\bf{\em{(ii)}}}~ In similarity-based filtering, the set of users/services to be filtered was sorted in ascending order based on their cosine similarity values with the target user/service. From the sorted list, the cosine similarity value (say, $\lambda^C_k$) of the $(\ceil{|\Gamma| \times k})^{th}$ user/service with the target user/service was chosen as the threshold (i.e., $T_s^u$/$T_s^s$), where $\Gamma$ is the set of users/services to be filtered. 
Finally, $\max(\lambda^C_k, \lambda^C_{max} \times k)$ was chosen as the value of $T_s^u$/$T_s^s$, where $\lambda^C_{max}$ is the maximum cosine similarity value across all users/services with the target user/service.
 
{{\bf{\em{(iii)}}}~ For context sensitive filtering, the threshold $T_{cs}^u$/$T_{cs}^s$ was chosen as $k^{th}$ fraction of the cardinality of the set of users/services obtained after similarity-based filtering.

It may be noted that $\lambda^H_{0.5}, \lambda^C_{0.5}$ represent the median values, and  
the value of $k$ was 0.5 in Equations \ref{eq:tr_cs_u}, \ref{eq:tr_cs_s}. 

\begin{figure}[!b]
    \centering
 	\includegraphics[width=0.85\linewidth]{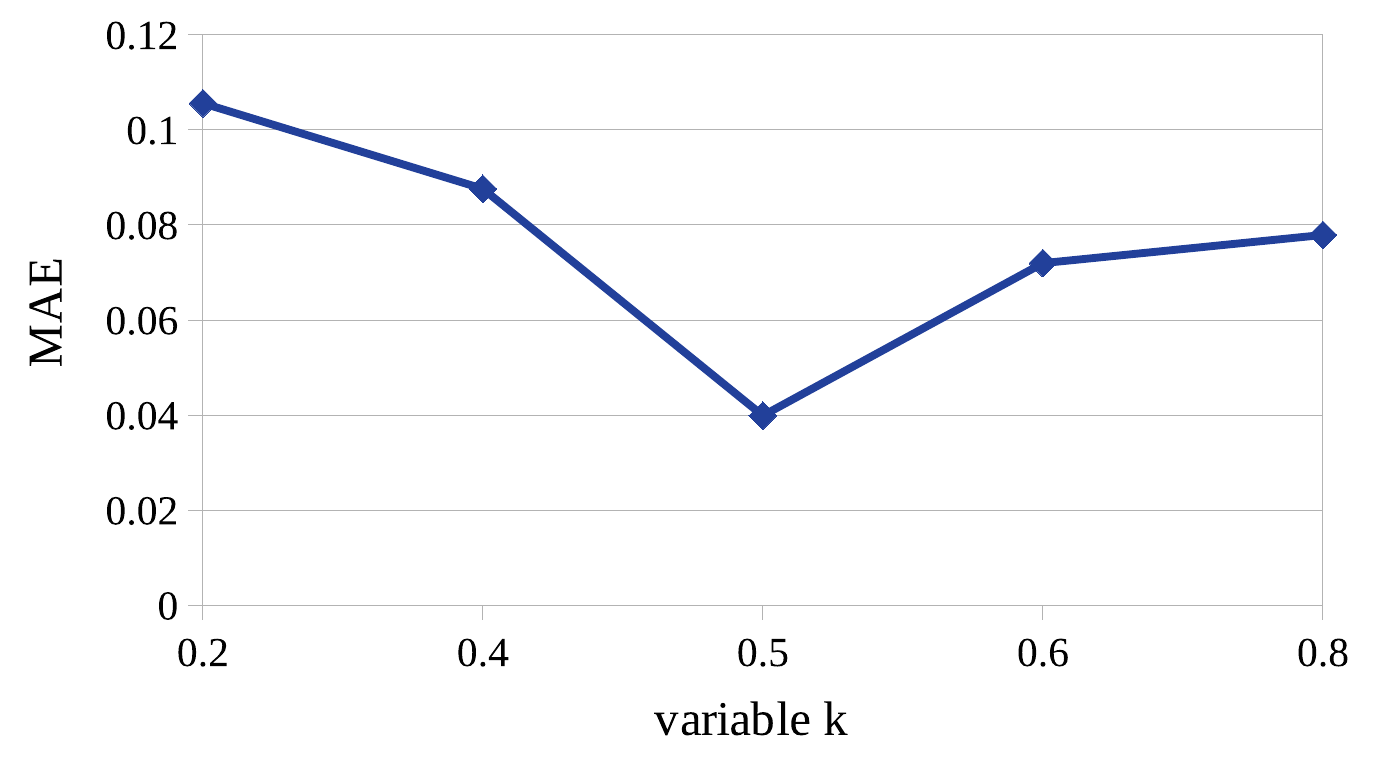}
   \caption{Change in MAE with respect to variable $k$, when the size of training dataset = 30\%}
 	\label{fig:tp_threshold}
\end{figure}

\noindent
Fig. \ref{fig:tp_threshold} shows the change in MAE with respect to $k$. 
As evident from Fig. \ref{fig:tp_threshold}, for $k = 0.5$, the least MAE value was obtained, which we used in our experiment. 
It may be noted, as we increased the value of $k$ beyond 0.5, the MAE value reduced, because of the following fact. 
Due to a high value of $k$, 
the values of the threshold parameters were also very high. 
Therefore, the filtered set of users/services contained only a few numbers of users/services, which was insufficient to train the neural network in the latter part of our framework. 
On the other hand, 
for lower values of the threshold parameters, 
the filtered set of users/services contained a large number of users/services, which in turn impeded the objective of the filtering. 
Consequently, MAE reduced with the decrease in the value of $k$.

{\emph{\ref{subsubsec:itp}.2}~ {\textbf{\emph{Impact of training data size of NRL-2:}}}
It may be observed, the size of the dataset to train NRL-2 (say, $TrD_{\text{NRL-2}}$) has an impact on the prediction accuracy. 
While Fig. \ref{fig:sca_level2_nn_data}(a) shows the change in MAE  with increase in $TrD_{\text{NRL-2}}$, Fig. \ref{fig:sca_level2_nn_data}(b) shows the computation time required by CAHPHF with the change in $TrD_{\text{NRL-2}}$. 
It may be noted further, the $x$-axis of Fig.s  \ref{fig:sca_level2_nn_data}(a), (b) represent 
the training data size of CAHPHF, whereas,
the legends of the figures 
represent the size of $TrD_{\text{NRL-2}}$. 
%
As observed from Fig.s \ref{fig:sca_level2_nn_data}(a), (b), the MAE value decreased with the increase in $TrD_{\text{NRL-2}}$ while compromising the computation time. Clearly, we have a trade-off between computation time and prediction accuracy. 
Furthermore, we observe from Fig.s  \ref{fig:sca_level2_nn_data}(a), (b) that 
even our worst performance  (i.e., the MAE value obtained by CAHPHF when $\left| TrD_{\text{NRL-2}}\right| = 50$) was better than the state-of-the-art approaches of Table \ref{tab:compareDS1}.
Therefore, as per the permitted time limit, the size of $TrD_{\text{NRL-2}}$ is to be decided. 

\begin{figure}[!b]
    \centering
 	(a)\includegraphics[width=0.9\linewidth]{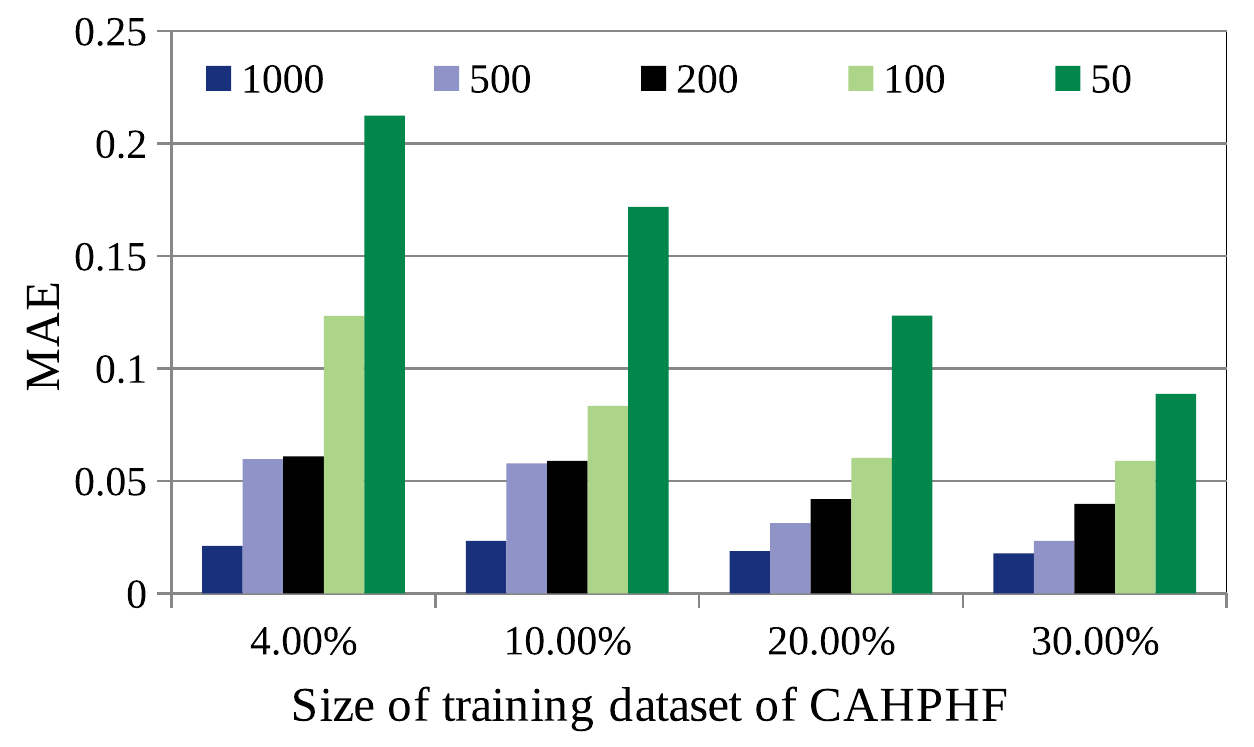}
 	(b)\includegraphics[width=0.9\linewidth]{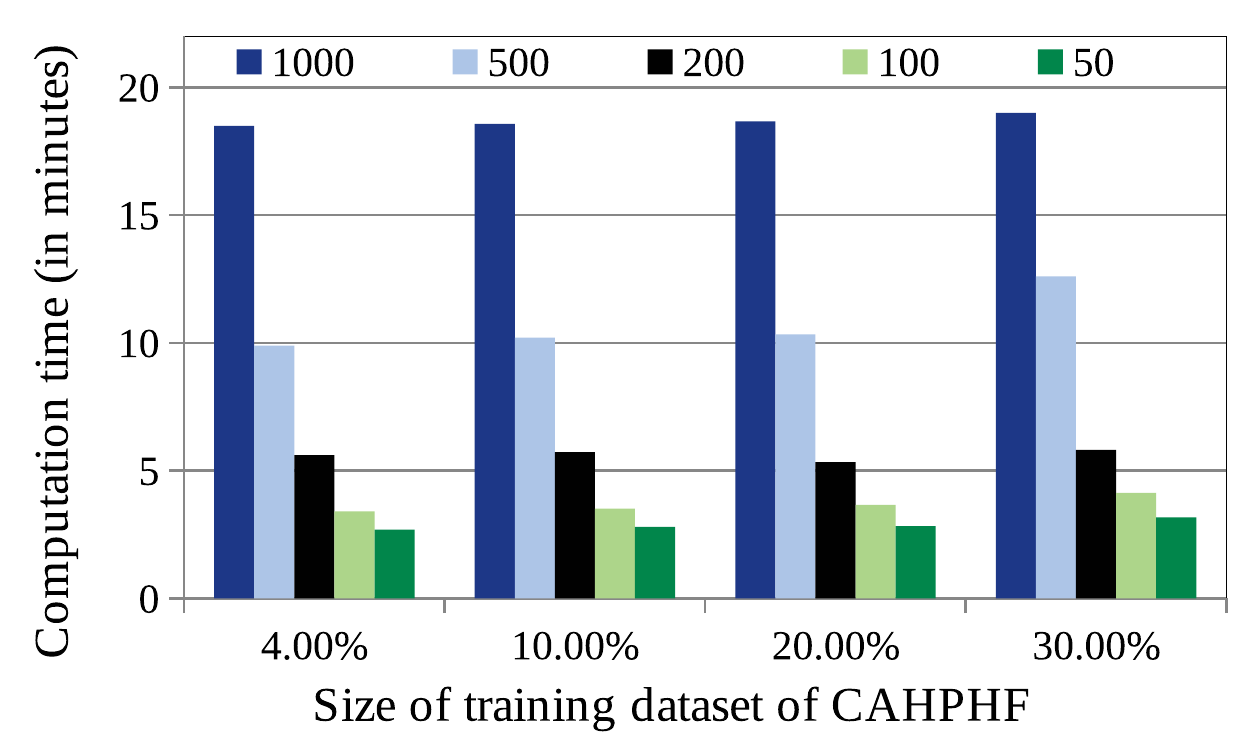}
\caption{Impact of training data size of NRL-2 on (a) MAE, (b) computation time}
 	\label{fig:sca_level2_nn_data}
\end{figure}

\begin{figure}[!t]
    \centering
 	(a)\includegraphics[width=0.8\linewidth]{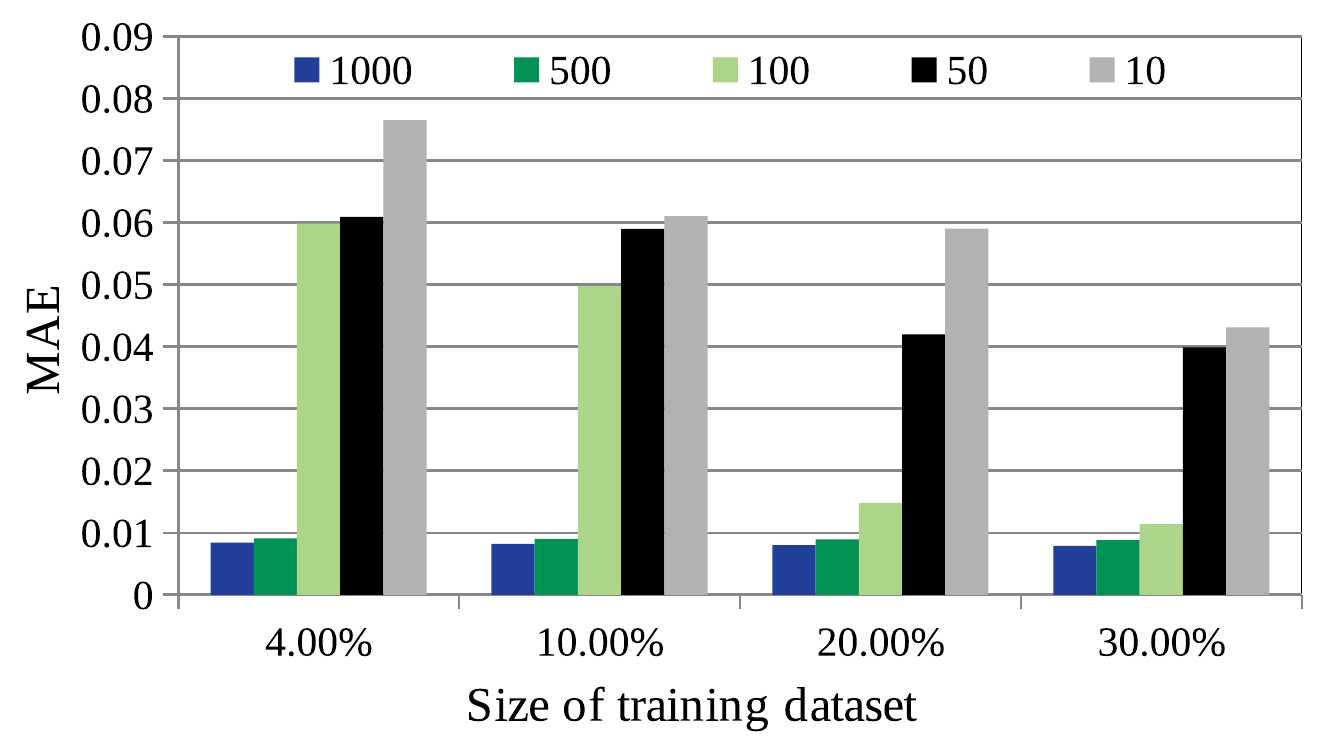}
 	(b)\includegraphics[width=0.8\linewidth]{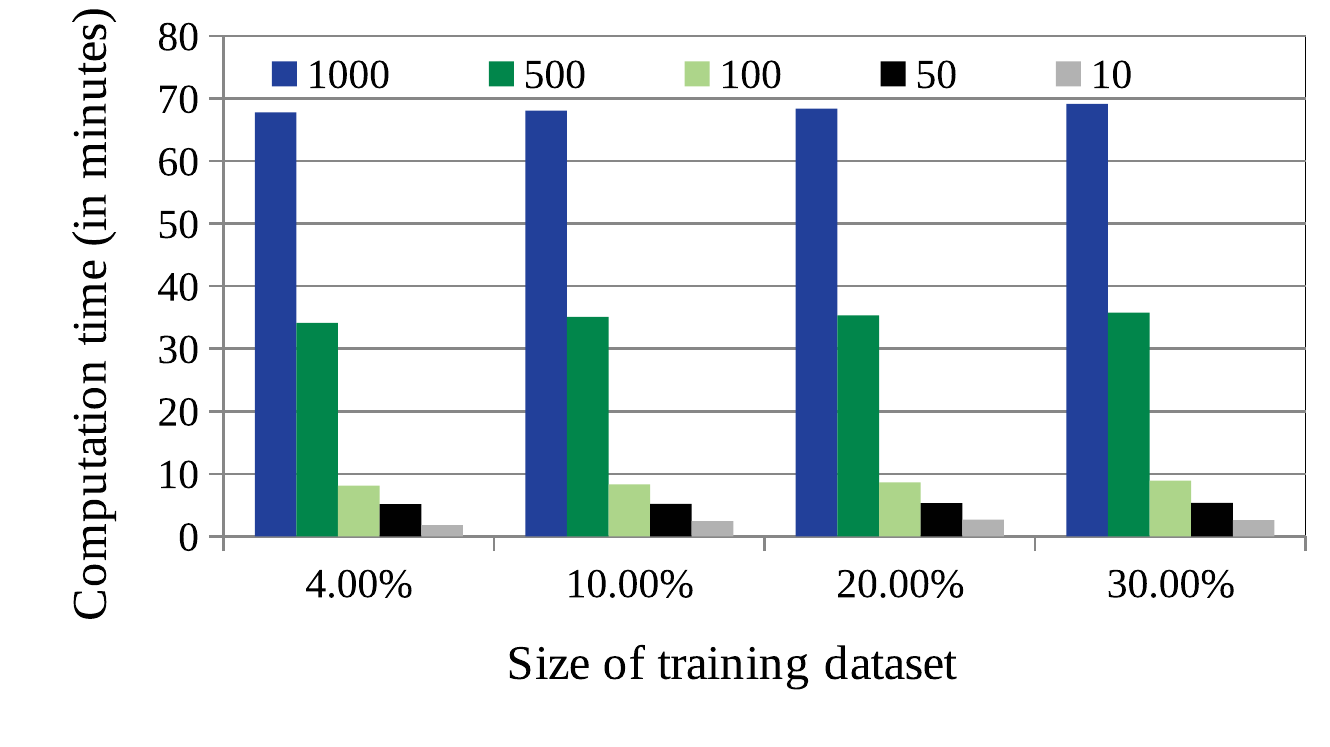}
\caption{Impact of the number of epochs of each NR of NRL-1 on (a) MAE, (b) computation time}
 	\label{fig:sca_level1_nn}
\end{figure} 
\begin{figure}[!t]
 	\centering
 	(a)\includegraphics[width=0.8\linewidth]{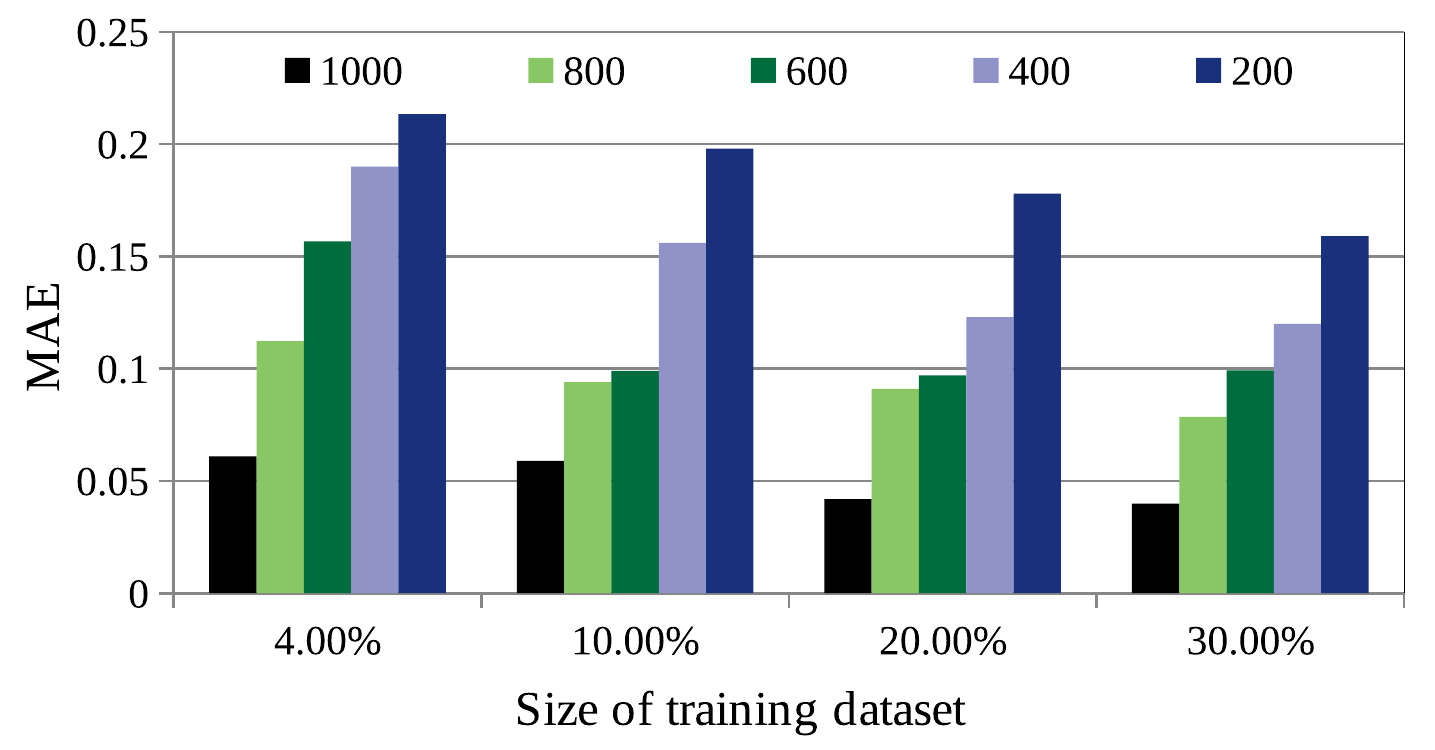}
 	(b)\includegraphics[width=0.8\linewidth]{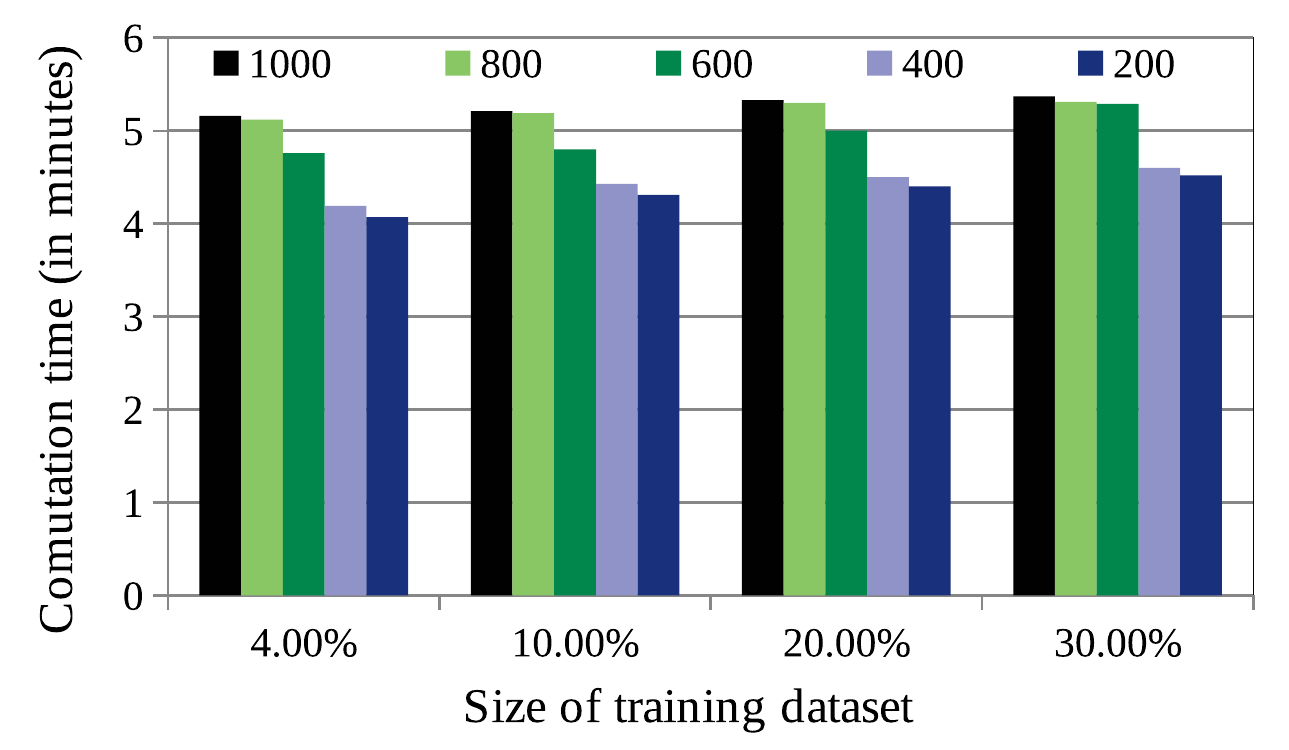}
   \caption{Impact of the number of epochs of NRL-2 on  (a) MAE, (b) computation time}
 	\label{fig:sca_level2_nn}
\end{figure}

{\emph{\ref{subsubsec:itp}.3}}~ {\textbf{\emph{Impact of Hyper-parameters of NR:}}}
Here, we mainly discuss about 3 hyper-parameters of the neural networks: 
(a) the number of epochs in each NR of NRL-1, 
(b) the number of epochs in NRL-2, and 
(c) the number of hidden layers of each NR in NRL-1. 
The other tunable parameters of the neural networks such as number of neurons in each hidden layer, the learning rate, momentum, minimal gradient value, etc. were also empirically decided.

Fig. \ref{fig:sca_level1_nn}(a) shows the change in MAE value with the increase in the number of epochs of each NR of NRL-1, while Fig. \ref{fig:sca_level1_nn}(b) shows the corresponding time to predict the QoS value by CAHPHF. 
It is observed from Fig.s \ref{fig:sca_level1_nn}(a), (b), the more time we spent to train the neural network, 
better was the solution quality. 
This statement is true at-least up to a certain number of epochs. 
Therefore, here also, we observed the time-quality trade-off. However, along with achieving a high prediction accuracy, 
{our framework should be robust. }
Therefore, considering the permitted time-limit, we need to choose the number of epochs of each NR of NRL-1. 
It may be noted, here also, 
in the worst case (i.e., the number of epochs of each NR of NRL-1 = 10), 
the MAE value obtained by CAHPHF was better than the state-of-the-art approaches of Table \ref{tab:compareDS1}.

Similar to the Fig.s \ref{fig:sca_level1_nn}(a), (b), 
we show the time-quality trade-off with respect to the number of epochs of NRL-2 in 
Fig.s \ref{fig:sca_level2_nn}(a), (b). 
The previous analysis is also valid here.

\begin{figure}[!t]
    \centering
 	\includegraphics[width=0.85\linewidth]{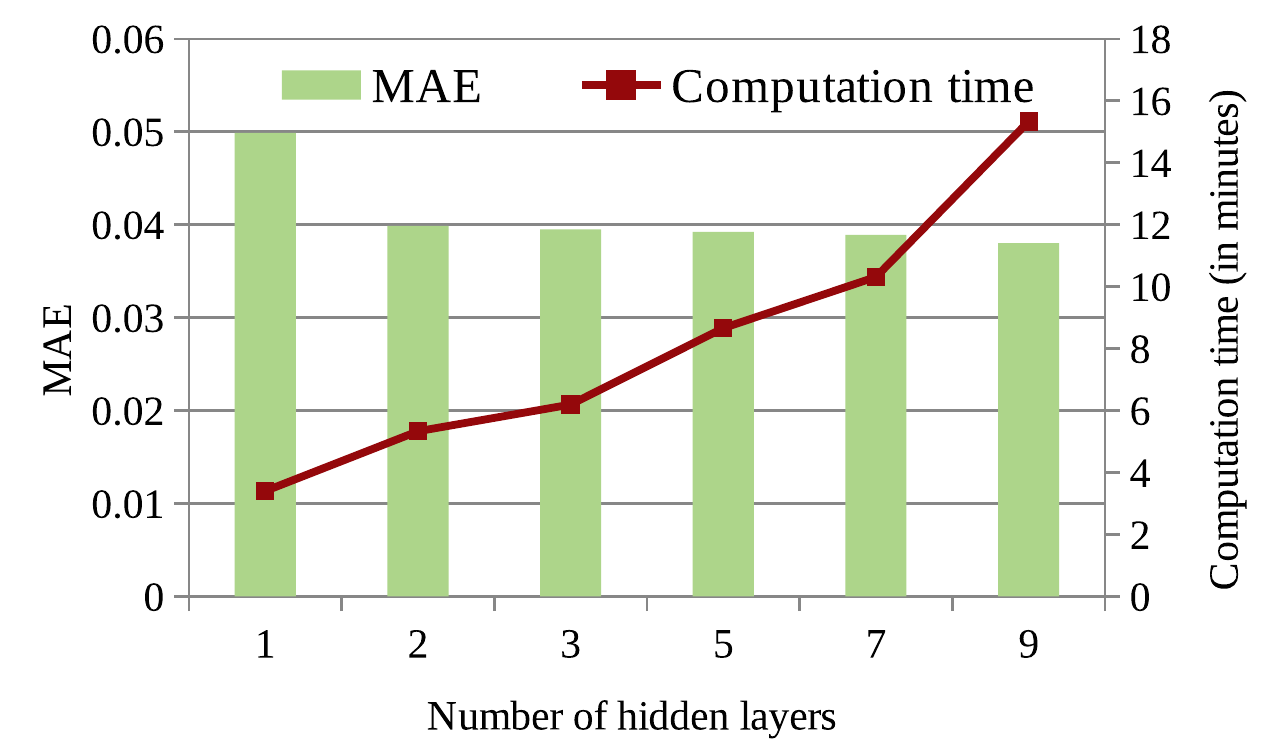}
\caption{Impact of the number of hidden layers of each NR of NRL-1 on (a) MAE, (b) computation time,  when the training data size = 30\%}
 	\label{fig:tp_hl}
\end{figure}

Fig. \ref{fig:tp_hl} shows the change in MAE value with respect to the number of hidden layers of each NR in NRL-1. 
The primary and secondary vertical axes of Fig. \ref{fig:tp_hl} represent MAE and computation time, respectively. 
As observed from this figure, the MAE value decreases with an increase in the number of hidden layers of the neural networks. However, with the increase in the number of hidden layers of the neural networks, computation time to generate the solution also increases. 
Therefore, here also, the number of hidden layers are chosen as per allowed time-limit.

In summary, our proposed CAHPHF, on the one hand, outperformed the state-of-the-art methods in terms of prediction accuracy, on the other hand, it generated the solution in a reasonable time limit.

\section{Related Work}\label{related}
\noindent
The QoS parameter plays a crucial role in various operations of the services life cycle, e.g., service selection \cite{DBLP:journals/access/AlayedDAMA19,DBLP:journals/access/DahanMA19,DBLP:journals/access/LiYGW19}, service composition \cite{DBLP:journals/tsc/ChenCC18,DBLP:journals/tsc/AminiO19,DBLP:journals/tweb/ChattopadhyayBB17}, service recommendation \cite{DBLP:conf/icsoc/ZouJNWPG18,DBLP:conf/IEEEscc/LiWSZ17}. 
The major limitation of most of these studies is that the service QoS values are assumed to be known.
However, the QoS value of service often changes across various factors, such as user \cite{DBLP:conf/icsoc/ZouJNWPG18,wu2017collaborative}, time \cite{DBLP:conf/IEEEscc/LiWSZ17}, location \cite{DBLP:conf/compsac/ChenMSXL17,DBLP:conf/IEEEscc/QiHSGL15,DBLP:journals/tsc/ZhengMLK13}, etc. 
Therefore, QoS prediction is an integral part of a service life cycle.


The QoS prediction has been studied broadly in the literature  \cite{DBLP:journals/tsc/LuoZWXZ19,DBLP:journals/access/YinZXZMY19,DBLP:conf/icws/XiongWLG17,DBLP:conf/icws/LeePB15,DBLP:conf/icws/WuQZWY15}.
The collaborative filtering is one of the major techniques to address the QoS prediction problem. The collaborative filtering can be of two different types: model-based and memory-based. The memory-based collaborative filtering \cite{sun2013personalized,wu2017collaborative,DBLP:conf/IEEEscc/LiWSZ17,DBLP:journals/soca/YuH16,DBLP:conf/icws/ZhouWGP15,DBLP:conf/icsoc/ZouJNWPG18} is further classified into two categories: user-based and service-based. In user-based collaborative filtering  \cite{DBLP:conf/icws/WuQWZY16,DBLP:journals/tsc/ZhengMLK13,breese1998empirical}, the similar users are taken into account to predict the QoS value, while in service-based collaborative filtering \cite{sarwar2001item}, the similar services are considered for QoS prediction. 
In this context, various similarity measures have been used to obtain the set of similar users or services, e.g., Pearson Correlation Coefficient (PCC) \cite{DBLP:conf/icws/ZhouWGP15}, cosine similarity measure \cite{DBLP:conf/icsoc/ChattopadhyayB19}, etc. 
Some enhanced similarity measures \cite{DBLP:conf/IEEEscc/LiWSZ17,DBLP:journals/soca/YuH16} have also been introduced to improve prediction accuracy. 
%
However, the stand alone user-based or service-based collaborative filtering may not be very effective concerning prediction accuracy, since it does not consider 
similar services (or, users) in a user-based (or, service-based) approach,  
and therefore, the prediction accuracy is not satisfactory. 
To achieve a higher prediction accuracy, both the techniques are combined further to predict the QoS values \cite{DBLP:conf/IEEEscc/LiWSZ17,DBLP:journals/soca/YuH16,DBLP:conf/icws/ZhouWGP15,DBLP:conf/icsoc/ZouJNWPG18}. However, the prediction accuracy of memory-based collaborative filtering approaches fall significantly for the sparse matrix.

\begin{table}[!t]\makegapedcells
\scriptsize
\caption{Brief literature review}
\centering
 \begin{tabular}{c V{3} c |c|c V{3} c|c V{3} c|c}
 \hline
 State-of- & \multicolumn{3}{c V{3}}{Methods} & \multicolumn{2}{cV{3}}{{QoS}} & \multicolumn{2}{c}{Dimension} \\
 \cline{2-8}
 the-Art & {\textbf{CF}} & {\textbf{MaF}} & {\textbf{Reg}} & \textbf{RT} & \textbf{TP} & {\textbf{Time}} & {\textbf{Location}}  \\ 
 \cline{5-6}
 \hline
 \hline
 \cite{DBLP:conf/icws/ZhouWGP15} & U + S &  &  & \checkmark & \checkmark &   &   \\ 
 \hline
 \cite{DBLP:conf/icsoc/ZouJNWPG18,wu2017collaborative,sun2013personalized} & U + S &  &  & \checkmark &  &  & \\ 
 \hline
 \cite{DBLP:journals/soca/YuH16} & U + S &  &  & \checkmark & \checkmark & \checkmark & \checkmark \\
 \hline
  \cite{DBLP:conf/compsac/ChenMSXL17} & U + S &  &  & \checkmark & \checkmark &   & \checkmark \\
 \hline
 \cite{DBLP:conf/icws/WuQWZY16} & U &  &  & \checkmark & \checkmark & \checkmark & \checkmark \\ 
 \hline
 \cite{DBLP:conf/IEEEscc/LiWSZ17} & U + S &  &  & \checkmark & \checkmark & \checkmark  &    \\ 
 
\hline
\cite{DBLP:journals/tsc/ZhengMLK13}& U & \checkmark &  & \checkmark & \checkmark &   & \checkmark  \\
\hline
\cite{DBLP:conf/wocc/LuoZXZ14} &  & \checkmark &  & \checkmark & \checkmark &   &  \\
\hline
\cite{DBLP:conf/colcom/Chen16a} &  &  & \checkmark & \checkmark & \checkmark &   &  \\
\hline
\cite{DBLP:conf/IEEEscc/QiHSGL15} & U + S & \checkmark &  & \checkmark &   &   & \checkmark \\
\hline
\cite{DBLP:conf/icws/XiongWLLH18} & U + S & \checkmark & \checkmark  & \checkmark & \checkmark & \checkmark  &  \checkmark \\
\hline
\cite{DBLP:conf/icsoc/ChattopadhyayB19} & U + S &  & \checkmark  & \checkmark & \checkmark &  &   \\
\hline
\cite{DBLP:conf/IEEEscc/ZhangSL0L16} &  &  & \checkmark & \checkmark & \checkmark & \checkmark &   \\
\hline
\cite{DBLP:journals/access/YinZXZMY19} & U + S &  \checkmark &  & \checkmark & \checkmark &  &   \\
\hline
\cite{breese1998empirical} & U &  &  & \checkmark & \checkmark &  &   \\
\hline
\cite{sarwar2001item} & S &  &  & \checkmark & \checkmark &  &   \\
\hline
 \cite{lo2012extended} & U + S & \checkmark &  & \checkmark &  &  &   \\
\hline
 \cite{wang2016multi} &  & \checkmark &  & \checkmark & \checkmark & \checkmark &   \\
\hline
\cite{DBLP:journals/access/ChenWXZZC19} &  & \checkmark &  & \checkmark & \checkmark & \checkmark & \checkmark  \\
\hline
\cite{DBLP:journals/access/GuoMCTH19} &  & \checkmark &  & \checkmark & \checkmark &  &   \\
\hline
\multicolumn{8}{r}{U: user-based; ~~ S: service-based; ~~Reg: regression}\\
\end{tabular}\label{tab:literature}
\end{table}

To resolve the problem with the sparse matrix and to improve the accuracy, a model-based collaborative approach is implemented. In the model-based collaborative filtering, a predefined model is adapted according to the given dataset. One such method is matrix factorization \cite{DBLP:conf/IEEEscc/QiHSGL15,DBLP:conf/wocc/LuoZXZ14,DBLP:journals/tsc/ZhengMLK13,DBLP:conf/icws/XiongWLLH18}, which is employed for solving the sparsity problem in memory-based collaborative filtering technique. The matrix factorization involves decomposing the QoS invocation log matrix into a low-rank approximation that makes further predictions. 
However, in the traditional matrix factorization method, the predicted value lacks accuracy.
Therefore, regularisation terms is included in the loss function of matrix factorization \cite{DBLP:conf/IEEEscc/QiHSGL15,DBLP:conf/wocc/LuoZXZ14,DBLP:journals/tsc/ZhengMLK13,lo2012extended} to avoid overfitting \cite{goodfellow2016deep} in learning process and  
to improve the prediction accuracy. 
To elevate this accuracy further, the collaborative filtering method is integrated with the matrix factorization, where the collaborative filtering utilizes the local information, and the matrix factorization uses global information for value prediction \cite{DBLP:conf/IEEEscc/QiHSGL15,DBLP:journals/tsc/ZhengMLK13,DBLP:conf/icws/XiongWLLH18}. 
Sometimes, the matrix factorization is combined with the recurrent neural network \cite{DBLP:conf/icws/XiongWLLH18} to obtain better results. 
In \cite{DBLP:conf/icws/XiongWLLH18}, a personalized LSTM (Long Short-Term Memory)-based matrix factorization is proposed to capture the temporal dependencies of both users and services to timely update prediction model with data. 
Some other approaches \cite{DBLP:conf/IEEEscc/QiHSGL15,DBLP:journals/tsc/ZhengMLK13} in the literature
have used a few other information (e.g., geographic location for neighborhood similarity calculation) in addition to the matrix factorization for improving accuracy level. On a similar note, for multi-dimensional QoS prediction problem, a tensor decomposition and reconstruction method \cite{wang2016multi} has been used for QoS values prediction. 
Although, matrix factorization can handle the sparsity problem, most of the times it suffers from loss of information \cite{DBLP:conf/wocc/LuoZXZ14}.

Another popular model-based technique is regression \cite{DBLP:conf/colcom/Chen16a}. In \cite{DBLP:conf/icsoc/ChattopadhyayB19,DBLP:conf/IEEEscc/ZhangSL0L16,DBLP:conf/icws/ShiZLC11}, the neural regression has been proposed to obtain better accuracy. 
Some neural network-based models exist in the literature to predict QoS value, e.g., back-propagation neural network \cite{DBLP:conf/colcom/Chen16a}, feed forward neural network \cite{DBLP:conf/icsoc/ChattopadhyayB19}, neural network with radial basis function \cite{DBLP:conf/IEEEscc/ZhangSL0L16}, etc. Predicting QoS values only on the basis of neural regression 
may not provide satisfactory outcome. 
Therefore, improvisation of the regression method has been introduced. For example, in \cite{DBLP:conf/icws/ShiZLC11}, a clustering of similar users on the basis of location along with the neural regression for QoS prediction has been proposed. In \cite{DBLP:conf/icsoc/ChattopadhyayB19}, a neural regression with filtering has been proposed, where a set of similar users and services are generated first, and then neural regression is employed to predict the QoS value. However, in \cite{DBLP:conf/icsoc/ChattopadhyayB19}, an ad-hoc approach is used to handle the sparsity problem. The authors in \cite{DBLP:journals/access/YinZXZMY19} has proposed QoS prediction with auto-encoder, where there is a 
model entitled that combines  
both model-based and memory-based approaches for QoS value prediction. However, the predicted QoS value is yet to reach the satisfactory level. 
Table \ref{tab:literature} provides a briefing on the reported works that addressed the prediction problem.

In contrast to the above approaches, our current paper addresses the QoS prediction problem by taking advantage of both the memory-based and model-based techniques. 
The proposed framework is of two-folds, hybrid filtering, followed by hierarchical prediction. Our filtering technique, on the one side, combines both user-based and service-based approaches, while on the other side, it is a coalition of user-intensive and service-intensive models to capture the priority of users and services. Moreover, to achieve better accuracy, our filtering technique leverages the location information of users and services. Furthermore, we apply a clustering technique to obtain a set of similar context-sensitive users and services. Our hierarchical prediction mechanism takes advantage of the model-based approach. To deal with the sparsity problem, we use collaborative filtering along with the matrix factorization to fill up the matrix. We then employ a hierarchical neural regression to predict the QoS value. On the one hand, our hierarchical neural regression fixes the QoS prediction problem. On the other hand, it helps in reducing the error in prediction. Our extensive experimental analysis also justifies the necessity of each segment of our framework.

\section{Conclusion}
\label{sec:conclusion}
\noindent
In this paper, we propose a hierarchical QoS prediction mechanism with hybrid filtering by leveraging the contextual information of users and services. Our approach takes benefit of the memory-based strategies by consolidating filtering and the model-based strategies by combining hierarchical prediction. Additionally, we handle the sparsity issue by filling up the absent values in the matrix using collaborative filtering and matrix factorization. Finally, we increase the prediction accuracy by aggregating the predicted values obtained by user-intensive and service-intensive modules using hierarchical neural regression. We performed extensive experiments on publicly available WS-DREAM benchmark datasets. The experimental results show that the proposed CAHPHF framework is better than the state-of-the-art approaches in terms of prediction accuracy, while justifying the requirement of each module of our framework. In the future, we will endeavor to work on a time-variant QoS prediction mechanism.

\bibliographystyle{IEEEtran}
\bibliography{ref}

\vspace{-0.5in}
\begin{IEEEbiography}[{\includegraphics[width=1in,height=1in,clip,keepaspectratio]{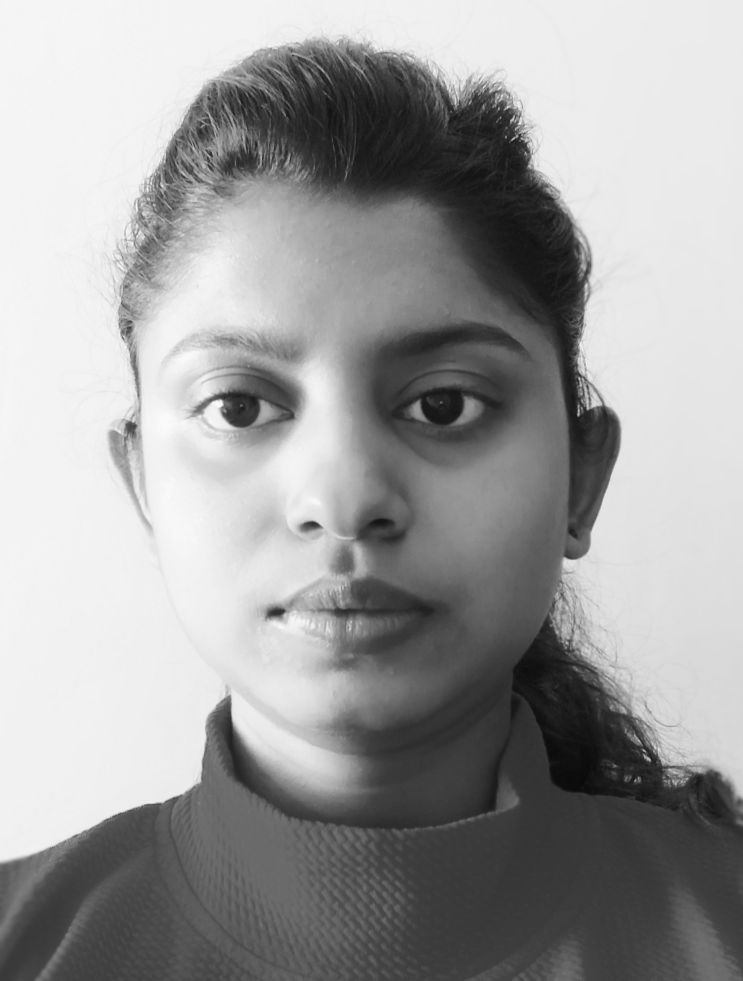}}]
{Ranjana Roy Chowdhury} received her B.Tech in Computer Science and Engineering (CSE) from Assam University, Silchar, India, in 2017. Currently, she is pursuing her M.Tech in CSE from Indian Institute of Information Technology Guwahati, India. Her research interests include Services Computing, Machine Learning-related subjects.
\end{IEEEbiography}
\vspace{-0.5in}
\begin{IEEEbiography}[{\includegraphics[width=1in,height=1in,clip,keepaspectratio]{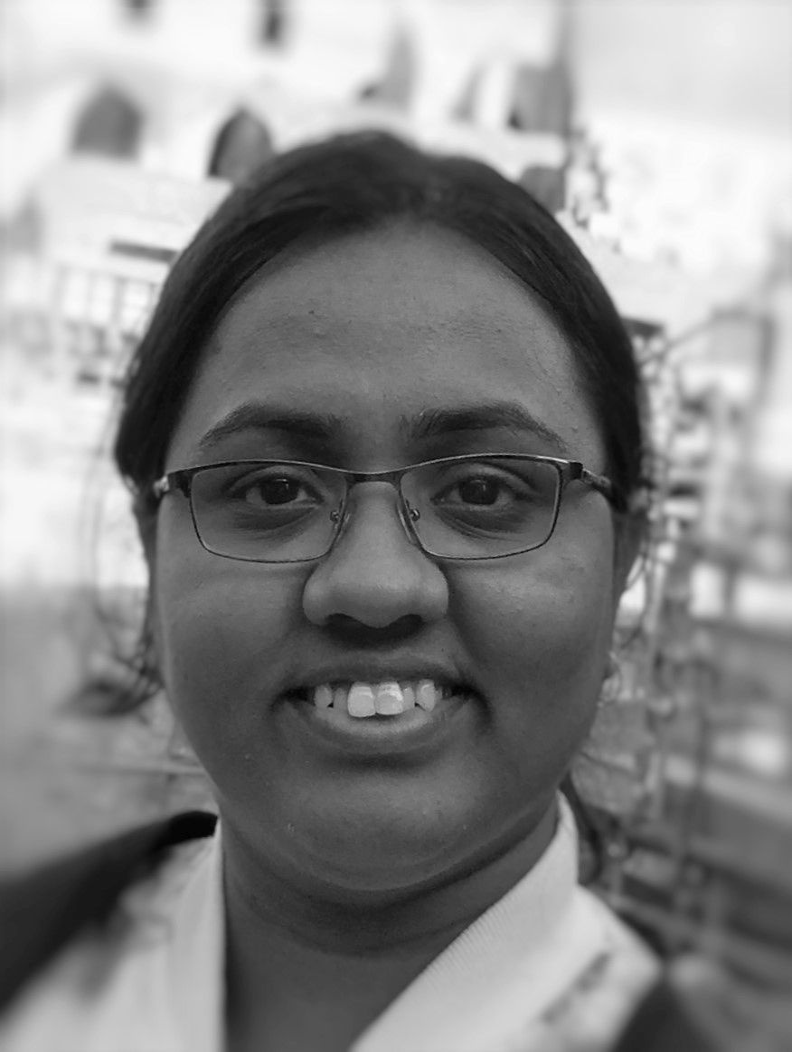}}]
{Soumi Chattopadhyay} (S'14, M'19) received her Ph.D. from Indian Statistical Institute, Kolkata in 2019. Currently, she is an Assistant Professor in Indian Institute of Information Technology Guwahati, India. Her research interests include Distributed and Services Computing, Artificial Intelligence, Formal Languages, Logic and Reasoning.
\end{IEEEbiography}

\vspace{-0.5in}
\begin{IEEEbiography}[{\includegraphics[width=1in,height=1in,clip,keepaspectratio]{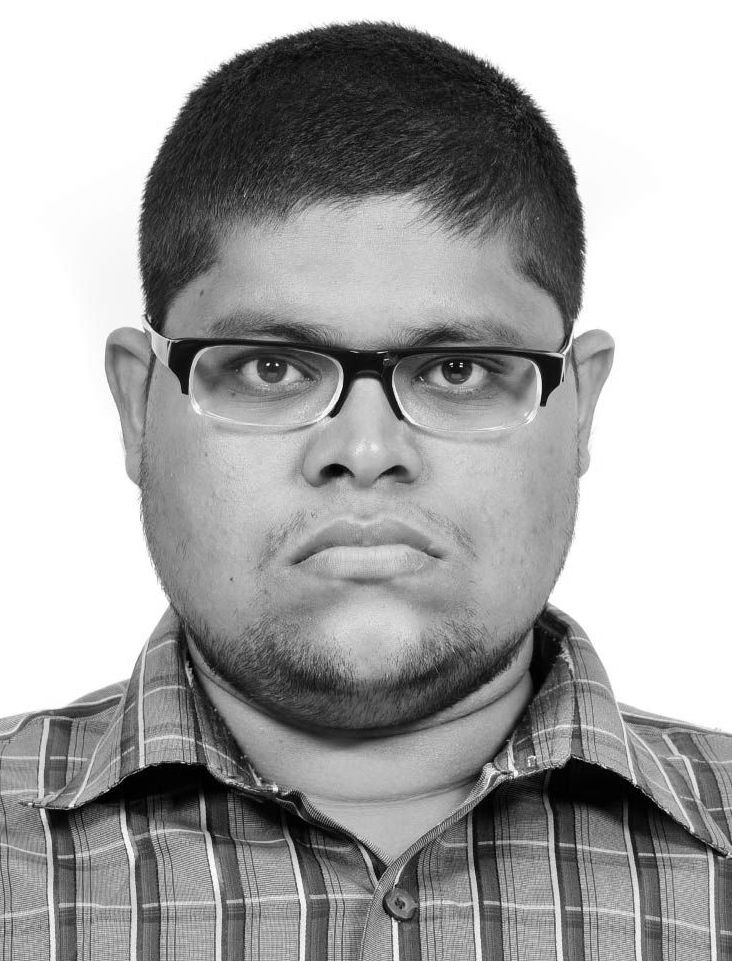}}]
{Chandranath Adak} (S'13, M'20) received his Ph.D. in Analytics from University of Technology Sydney, Australia in 2019. Currently, he is an Assistant Professor at Centre for Data Science, JIS Institute of Advanced Studies and Research, India. His research interests include Image Processing, Pattern Recognition, Document Image Analysis, and Machine Learning-related subjects.
\end{IEEEbiography}

\end{document}